\documentclass{amsart}

\usepackage{times}
\usepackage{amsfonts, amsmath, amssymb, amscd, mathrsfs}
\usepackage{color,overpic,wasysym,xcolor}
\usepackage{graphicx,pgf}
\usepackage{subfig}
\usepackage{pstricks,pst-grad,pst-plot,pst-node}
\usepackage{tikz}
\usetikzlibrary{arrows}
\usepackage{epic}
\usepackage[section]{algorithm}
\usepackage{algorithmic}
\usepackage{amsaddr}

\newtheorem{theorem}{Theorem}[section]

\newtheorem{proposition}[theorem]{Proposition}

\newtheorem{assumption}[theorem]{Assumption}

 \theoremstyle{definition}
\newtheorem{example}[theorem]{Example}
 \newtheorem{definition}[theorem]{Definition}

 \theoremstyle{plain}

\def\set#1{\{ #1 \}}

\def\N{\mathbb{N}}
\def\R{\mathbb{R}}

\def\opt{{\mathit{opt}}}
\def\Post{{\mathit{Post}}}

\begin{document}

\title{A theory of robust software synthesis}
\author{Rupak Majumdar$^1$$^2$, Elaine Render$^1$ and Paulo Tabuada$^1$} 
\address{$^1$ UC Los Angeles, $^2$Max Planck Institute for Software Systems}
\email{rupak@mpi-sws.org, elaine@cs.ucla.edu,tabuada@ee.ucla.edu} 
\thanks{This research was funded in part by the NSF awards 0834771, 0953994, and 1035916.}
\markleft{RUPAK MAJUMDAR ET AL.}

\begin{abstract}
A key property for systems subject to uncertainty in their operating environment is \emph{robustness}, ensuring that unmodelled, but bounded, disturbances have only a proportionally bounded effect upon the behaviours of the system.  Inspired by ideas from robust control and dissipative systems theory,
we present a formal definition of robustness and algorithmic tools for the design of optimally robust controllers
for $\omega$-regular properties on discrete transition systems.
Formally, we define \emph{metric automata} ---automata equipped with a metric on states---
and strategies on metric automata which guarantee robustness for $\omega$-regular properties.
We present fixed point algorithms to construct optimally robust strategies in polynomial time.
In contrast to strategies computed by classical graph theoretic approaches, the strategies
computed by our algorithm ensure that the behaviours of the controlled system gracefully degrade under the action of disturbances; the degree of degradation is parameterized by the magnitude of the disturbance.
We show an application of our theory to the design of controllers that tolerate infinitely many
transient errors provided they occur infrequently enough.
\end{abstract}

\maketitle

\definecolor{yellow}{cmyk}{0,0,0.9,0.0}


\section{Introduction}

Reactive software systems that respond directly or indirectly to information coming from an uncertain environment
are a fundamental component of many mission-critical applications--- in healthcare, energy-distribution,
and industrial automation--- with enormous societal impact.
It is widely recognized that the current design and verification methodologies fall short of what is required
to design these systems in a robust yet cost-effective manner.

Current approaches to system design and verification are only able to differentiate between absolutely correct
behaviour and incorrect behaviour, providing no way of quantifying precisely the effects of errors.
Hence a catastrophic failure is indistinguishable from a small deviation and no guarantees as to the
resulting effects on the nominal system behaviour may be made.
Clearly, this view is overly restrictive.
First, reactive systems need to operate for extended periods of time in environments that are either unknown
or difficult to describe and predict at design time.
For example, sensors and actuators may have noise, there could be mismatches between the dynamics
of the physical world and its model, software scheduling strategies can change dynamically.
Thus, asking for an environment and program model that encompasses all possible
scenarios places an undue burden on the programmer, and the detailed book-keeping
of every deviation from nominal behaviour renders the specifications difficult to understand and maintain.
Second, even when certain assumptions are violated at run-time, we would expect the system to
behave in a \emph{robust} way: either by continuing to guarantee correct behaviour or by ensuring that the
resulting behaviour only deviates modestly from the desired behaviour under the influence of small perturbations.
Unfortunately, current design methodologies fall short in this respect: since the effects of errors cannot be explicitly quantified
no guarantees may be made that small changes in the physical world, in the software world, or in their interaction,
still result in acceptable behaviour.

In this paper, we present a theory of robustness for systems modelled discretely using \emph{automata}.
We are inspired by the well developed notion of robustness in continuous control theory, and the
tools and methodologies which have been successfully applied therein.
In the continuous world, the designer specifies the control system for the nominal case,
ignoring the potential effects of errors on system behaviour and performance.
The design methodology is such that guarantees may then be made as to the degree of degradation
of functionality of the controlled system under disturbances of bounded power.
We aim to provide a similar theory and algorithmic tools in the presence of
discrete changes on the one hand, and in the presence of more complex
temporal specifications ---given, for example, in linear temporal logic (LTL) or as $\omega$-automata---
on the other hand.
We do this in three steps.

First, robustness is a topological concept. In order to define it,
we need to give meaning to the word ``closeness'' or ``distance.''
For this, we define a metric on the system states.
Second, instead of directly modeling the effect of every disturbance, we model a \emph{nominal} system
(the case with no disturbance), together with a set of (unmodeled) disturbances whose effect can be bounded
using the metric.
That is, while making no assumption on the nature or origin of disturbances, we assume that
the disturbances can only push the system to a state within a distance $\gamma$ of the nominal state.
Third, under these three assumptions, we show how we can derive strategies for $\omega$-regular
objectives that are robust in that the deviation from nominal behaviour can be bounded as a function
of the disturbance and the parameters of the system.

To illustrate this last point, consider \emph{reachability properties} $\Diamond F$, where the system tries to reach
a given set of states $F$.
We provide fixed point algorithms which compute strategies that ensure $F$ is reached in the nominal case, and additionally,
when disturbances are present of magnitude $\gamma$, guarantee that the system reaches a set $F'$
which contains states at a distance of $\sigma \gamma$ or less from $F$, where $\sigma \in \mathbb R_0^+$.
Hence we may regard $\sigma$ as a measure of robustness of this strategy.
We also provide guarantees that the resulting inflation in the size of the acceptance set is indeed optimal.
Additionally, we show that an arbitrary strategy obtained through classical automata-theoretic constructions
(e.g., \cite{McNaughton,Zielonka98}) may provide trivial robustness guarantees
(e.g., a bounded disturbance can force the system to reach any arbitrary state).
We show how similar arguments can be made to provide robustness bounds for B\"uchi and parity (and thus, for all LTL)
specifications under the presence of disturbances.

Technically, our constructions lift arguments similar to arguments in robust control
based on control Lyapunov functions to the setting of $\omega$-regular properties.
For reachability, the correspondence is simple: we require that the strategy decrements a ``rank function''
at a rate that depends on the distance to the target.
For parity, the argument is more technical, and uses progress measures for parity games \cite{klarlund90,Namjoshi01}.
Finally, we provide simple fixed point methods to compute optimal robustness bounds and strategies attaining these bounds.

We also consider a simple application of our theory to the synthesis problem in the presence of transient
faults.
We show how, using our methodology, we can algorithmically synthesize controllers for LTL objectives
which provide a time-space tradeoff whereas classical automata-theoretic techniques may not be able to provide
any such property without first explicitly modelling all of the parameters of the fault in detail.

\smallskip
\noindent
\textbf{Related work.}
The work presented here is inspired by the theory of robust continuous control \cite{L2Gain,ROBUSTCONTROL_TEXTBOOK} and the theory of
infinite games with $\omega$-regular objectives on discrete graphs \cite{EmersonJutla91,Thomas95,Zielonka98}.
There has not been much previous work combining robustness with automata theory.
In discrete control systems, tolerance to errors is achieved by explicitly modeling faults and then solving a game assuming that
the adversary determines when faults occur \cite{Girault09}.
As mentioned earlier, the enumeration of possible faults can be tedious, if not impossible, at design time.
Topologies for hybrid systems \cite{HSTopology,HSModels} have been examined before, but the interactions with $\omega$-regular
specifications have not been.

{\em Qualitative} notions of fault tolerance have been studied in distributed systems, for example, by designing algorithms
to be ``self-stabilizing'' on perturbations \cite{Dijkstra74}, or by
requiring that an invariant
is eventually restored after an error (``convergence'') or that the system satisfied a more liberal invariant under an error (``closure'') \cite{Arora93}.
However, quantitative notions, relevant to discrete systems, have not been studied.
Our synthesis procedure produces strategies which satisfy quantitative notions of closure and (under some assumptions on the rate of faults) convergence.

In a series of papers \cite{Bloem091,Bloem09,Bloem10,Cerny10},
robustness measures are developed by comparing the number of environmental errors and the number
of resulting system errors using cost functions.
Bloem et al.\ \cite{Bloem09} define \emph{$k$-robustness}: roughly speaking,
a system is $k$-robust if the ratio of system to environment errors is $k$.
In general terms, the synthesis approach presented here
results in $\sigma$-robust strategies for constant disturbance bounds, where $\sigma$ is a constant associated with the rank
function which serves as a formal characterization of the strategy.
We demonstrate methods to construct strategies with optimal $\sigma$ values.
Moreover, we work in a simpler model, where the only adversarial action is the bounded disturbance,
while the work of Bloem et al.\ considers an explicit adversary.
Our framework has the advantage of leading to simple polynomial-time algorithms for synthesis, but may provide more
conservative results than the game-solving algorithms from \cite{Bloem10} when robustness is sought
in the presence of explicit adversaries.
A more detailed technical comparison with their work is provided in Section~\ref{safetysec}.

Tarraf et al.\  \cite{Tarraf08} develop a framework for quantifying robust stability in
finite Mealy machines by extending classical notions of gain stability from control theory.
The focus is upon input-output stability and, although we adopt a state-space approach,
the results derived in this paper for reachability are similar.
However, we go beyond simple reachability properties and consider also B\"uchi and parity requirements.
A more technical comparison appears in Section~\ref{safetysec}.

Measures of robustness against transient fault models have been studied in
the context of combinational circuits and FPGAs \cite{golshan07:dac,Hu-iccad08,Krishnaswamy-tcad09,Miskov-ZivanovM10}, but
extensions to temporal behaviours have not been considered.

\section{Preliminaries}\label{prelim}

Let $Q$ be a (finite or infinite) set. A function $d: Q \times Q \to \mathbb R_0^+$ is called a \emph{metric} or \emph{distance function} for $Q$ if for all $p, q, r \in Q$,
we have
\begin{itemize}
\item[(i)] $d(p,q) = 0$ if and only if $p = q$ (\emph{identity of indiscernibles});
\item[(ii)] $d(p, q ) = d(q, p)$ (\emph{symmetry});
\item[(iii)] $d(p,r) \leq d(p,q) + d(q,r)$ (\emph{triangle inequality}).
\end{itemize}
The pair $(Q,d)$ where $Q$ is a set and $d$ is a metric for $Q$ is called a \emph{metric space}.
Using a metric $d$ we define the distance from a single state $q\in Q$ to a set of states $Q' \subseteq Q$ as \mbox{$d(q,Q') = \inf_{q' \in Q'} d(q,q')$,} the shortest distance from $q$ to some element of the set $Q'$.

A function $R:Q \to \mathbb R$ is \emph{Lipschitz continuous} if there exists some constant $K > 0$ such that for any two states $q, q' \in Q$:
\[
|R(q) - R(q')| \leq K d(q,q'),
\]
that is, the absolute value of the difference between the images of $q$ and $q'$ is bounded above by a constant multiple of the distance between $q$ and $q'$ for every pair of states in $Q$.
The value $K$ is called the \emph{Lipschitz constant} of the function $R$ with respect to the distance $d$.
Note that if the set $Q$ is finite then every real valued function of $Q$ has this property.

We model discrete control systems using \emph{automata}.
Intuitively, we consider a ``nominal'' automaton modeling the undisturbed
dynamics of the system, and add a set of disturbance actions which can
perturb the nominal behaviour.
We consider a very general model for the disturbances by simply requiring their effects to be bounded but otherwise arbitrary.

For a set $\Sigma$ of symbols we let $\Sigma^*$ represent the set of finite strings of symbols from $\Sigma$, and let $\Sigma^{\omega}$ denote the set of infinite strings over $\Sigma$; we let $\lambda$ denote the empty string. 
The notation $|\Sigma|$ represents the cardinality of the set $\Sigma$ and $\Sigma^+$ is the set of non-empty finite strings over $\Sigma$.
A (metric) \emph{automaton} is a tuple \mbox{$A = ((Q,d), q_0, \Sigma, X, \delta, \gamma)$,}
where
\begin{itemize}
\item $Q$ is a set of \emph{states} and $(Q,d)$ is a metric space;
\item $q_0 \in Q$ is the unique \emph{initial state};
\item $\Sigma$ is a set of (system) \emph{input actions};
\item $X$ is a set of \emph{disturbance indices} including a special symbol $\epsilon$ signifying ``no disturbance'';
\item $\delta: Q \times \Sigma \times X \to Q$ is the \emph{transition function} specifying
the next state given the current
state, the input letter chosen by the system and some member of $X$
chosen by the environment and finally
\item $\gamma:Q \to \mathbb R_0^+$ is a real-valued function such that
for each $p \in Q$ and for every $a \in \Sigma$ such that $\delta(p,a,\epsilon)=q$ for some $q \in Q$
\[
d( q, \delta(p,a,x)) \leq \gamma(q) \ \textrm{for every} \  x \in X.
\]
\end{itemize}

Note that the disturbance bound is defined with respect to the \emph{target} state of a given transition, and not the \emph{source} state (that is, the inequality above is bounded by $\gamma(p)$ and not $\gamma(q)$). It would be a straightforward matter to reformulate the results herein with $\gamma(p)$ replaced by $\gamma(q)$.

An automaton is \emph{finite} if $Q$, $\Sigma$, and $X$ are all finite sets.
For an automaton $A$,
we define the \emph{undisturbed} or \emph{nominal} automaton, written $A_{\epsilon}$, as the automaton resulting from
restricting the set of disturbance indices $X$ to the singleton $\set{\epsilon}$.
For $q \in Q$, $a \in \Sigma$ and $x\in X$ we use the shorthand $q^{ax}$ to denote the state $\delta(q,a,x)$. We let $\overline{\gamma} = \sup_{q \in F} \gamma(q)$. If $\gamma(q) = \gamma(p)$ for all $p,q \in Q$ we say that $A$ has \emph{constant disturbance bound} and hence $\gamma(q) = \overline{\gamma}$ for all $q \in Q$.

Intuitively, the undisturbed automaton models the ``nominal'' behaviour of an automaton,
and the set of disturbance indices $X$ models possible environmental disturbances to the nominal behaviour (the symbol $\epsilon\in X$
thus represents the case where there is no disturbance).
The function $\gamma$ limits the effects of the disturbances with respect to the nominal behaviour at each state $q$:
when an action $a$ is chosen, the disturbances can cause a state
at most distance $\gamma(q^{a\epsilon})$ away from the nominal state to be reached instead.
As a special case, if $\gamma(q^{a\epsilon}) = 0$, then the disturbances have no effect on the nominal behaviour
(i.e., $q^{ax} = q^{a\epsilon}$ for each $a\in \Sigma$, and $x\in X$).

A \emph{trace} $\tau \in Q^* \cup Q^{\omega}$ of the automaton $A$ is a (finite or infinite) sequence of states
$\tau = q_0 q_1 q_2 \ldots$ from $Q$ such that
$q_0$ is the initial state of the automaton, and
there exist inputs $a_0, a_1, a_2, \ldots$ and disturbances $x_0, x_1, x_2, \ldots$ with $\delta(q_i, a_i,x_i) = q_{i+1}$ for $i \geq 0$.
For $q\in Q$ we write $q \in \tau$ and say the state $q$ appears on the trace $\tau$ if $q = q_i$ for some $i\geq 0$.
A \emph{nominal trace} is a trace in the nominal automaton $A_{\epsilon}$, that is, $\tau$ is such that $x_i = \epsilon$ for all $i \geq 0$.
For a finite trace \mbox{$\tau = q_0q_1\ldots q_n \in Q^*$,} we define $|\tau| = n + 1$, the \emph{length} of $\tau$.


The proposed model for the disturbances encompasses a wide range of concrete applications ranging from
the discrete to the continuous world, as illustrated by the next two examples.

\begin{example}{\bf [Digital Design with Single Bit-Flips]}\label{exdisc}
Consider an automaton $A$ modeling a state machine whose states are
encoded using a binary Gray code \cite{DigitalDesignTextbook}.
Each state of $A$ is a sequence of $n$ bits, and neighbouring states differ in only one bit.
Disturbances occur as single-event upsets which can cause a
single bit in the state to flip.
The distance function for the automaton is defined to be the Hamming
distance between $n$-bit strings.
The set of disturbance actions $X \subset \set{0,1}^n$ contains all binary strings of length $n$ with at most one non-zero
digit.
Under this definition, $\epsilon$ is equal to the binary string of length $n$ consisting
entirely of zeros.
The transition function $\delta$ for $A$ is defined from the transition function $\delta_\epsilon$ of $A_\epsilon$ by $\delta(q,a,x) = \delta_\epsilon(q,a) \oplus x$
for any $q \in \set{0,1}^n$ where $\oplus$ is the XOR function.
Hence, the potential effect of the disturbance is bounded by the constant $\overline{\gamma} = 1$.
%
\end{example}

\begin{example}{\bf [Robust Control]}
Consider a continuous control system in discrete time which may be viewed as an infinite-state automaton with transition
function $\delta:\R^n\times \R^m\times\R^p\to \R^n$. The state set is $\R^n$, the input alphabet is $\R^m$ and $\R^p$ is the set of environmental disturbances.
Disturbance signals $x:\N\to \R^p$ are often used as a lumped representation for several sources of uncertainly such as
measurement errors or errors in the model of the transition function.
Hence, the disturbance signals are assumed to be arbitrary but of bounded amplitude, that
is, $\Vert x(k)\Vert\le \gamma'$ for some constant $\gamma'\in \R_0^+$, some
norm $\Vert\cdot\Vert$ on $\R^p$, and every $k\in\N$.
A further typical assumption is Lipschitz continuity of $\delta$.
It then follows from these two assumptions that \[\Vert \delta(q,a,x)-\delta(q,a,0)\Vert\le K'\Vert x-0\Vert\le K'\gamma'\] where $K'$ is the Lipschitz constant.
Therefore by defining the distance function $d$ as \mbox{$d(y,z)=\Vert y-z\Vert$} we conclude that the system in this example has constant disturbance bound $\overline{\gamma} \in \R_0^+$ equal to $K'\gamma'$.
\end{example}

We make certain natural assumptions as to the connectedness of the automata we consider. In order to elucidate these assumptions we define the following notions.
A state $q \in Q$ is \emph{(nominally) reachable} if there exists a finite (nominal) trace connecting $q_0$ to the state $q$,
and \emph{(nominally) coreachable} with respect to some set of states $Q' \subseteq Q$ if there exists a finite (nominal) trace connecting $q$ to some state in $Q'$.
If every state in $Q$ is reachable (resp. coreachable w.r.t.\ $Q'$) we say that $A$ is reachable (resp., coreachable w.r.t.\ $Q'$).
Throughout the following we will assume that every automaton we consider is reachable.

We associate \emph{acceptance conditions} with automata to distinguish between ``good'' and ``bad'' traces.
A \emph{reachability condition} is a set $F \subseteq Q$ of \emph{terminal states}.
A \emph{reachability automaton} $(A,F)$ consists of an automaton $A$ together with
a reachability condition $F$.
A finite trace of the automaton $A$ satisfies the reachability condition $F$ if and only if it
ends at some state in the set $F$. We make the following assumption for all reachability automata in the paper.

\begin{assumption}\label{connectreach}
The automaton $A$ is nominally coreachable with respect to $F$.
\end{assumption}

A \emph{B\"uchi automaton} $(A, F)$ is an automaton $A$ together with a \emph{B\"uchi acceptance condition} $F \subseteq Q$.
For an infinite trace $\tau =q_0q_1\ldots\in Q^{\omega}$ let \[\zeta(\tau) = \set{q\in Q\mid \forall i\geq 0 \ \exists j>i,q_j = q}\]
denote the set of states appearing infinitely often on $\tau$.
A trace $\tau \in Q^{\omega}$ satisfies the B\"uchi acceptance condition $F$ if and only if $\zeta(\tau) \cap F \not = \emptyset$.
In other words, there exists at least one state in the set $F$ which features infinitely often on the trace. We again make Assumption \ref{connectreach}.

A \emph{generalized B\"uchi acceptance condition}
is a set of the form $\mathscr F = \{ F_0, \ldots F_{n-1} \}$ and consists of a finite number of subsets of the state set $Q$.
An automaton $A$ paired with such an acceptance condition
is called a \emph{generalized B\"uchi automaton}.
An infinite trace $\tau \in Q^{\omega}$ of $A$ satisfies the acceptance condition $\mathscr F$ if and only if $\zeta(\tau) \cap F_i \neq \emptyset$
for $i = 0, \ldots , n-1$. We ask that the following assumption is satisfied.
\begin{assumption}\label{connectgenbuchi}
The generalized B\"uchi automaton $A$ is nominally coreachable with respect to $F_i$ for $i=0,\ldots,n-1$.
\end{assumption}
The justifications for this assumption will be discussed further in Section \ref{subgenbu}.

Finally a \emph{parity automaton} $(A, \mathscr F)$ is an automaton $A$ together with a
\emph{parity acceptance condition} consisting of a finite number of pairwise disjoint subsets of the state set $Q$: $\mathscr F = \{F_1, \ldots, F_{2n+1}\}$
with $F_i \cap F_j = \emptyset$ for $i \neq j$.
The \emph{parity} of a state $q \in Q$ is the index $i$ of the unique set
$F_i$ containing $q$, if any, and undefined if there exists no such $F_i$.
A trace $\tau \in Q^{\omega}$ of $A$ satisfies the acceptance condition $\mathscr F$ if and only if the least parity amongst the states in the set $\zeta(\tau)$ is even.
Note that we allow the set of states to be partially colored \cite{Zielonka98}: the set $\cup_{i=1}^{2n+1} F_i$ may not necessarily cover $Q$. The connectedness assumption for parity automata is the following.

\begin{assumption}\label{connectparity}
Each state $q \in Q$ in the parity automaton $A$ is nominally coreachable with respect to some set of even parity $F_{2i}$, and if $q$ has odd parity,
then it is nominally coreachable with respect to some $F_{2i}$ where $2i$ is less than the parity of $q$.
\end{assumption}

The reasoning for this assumption will be given in Section \ref{subparity}.

A \emph{strategy} for an automaton $A$ is a function $S : Q^+ \to 2^{\Sigma}$ specifying input choices for each finite trace.
Given a strategy $S$, the set of \emph{outcomes} is the set of traces $q_0 q_1\ldots$ on which
$q_0$ is the initial state of the automaton, and for each $i\geq 0$ there exists $a_i \in S(q_0\ldots q_i)$ with $q_{i+1} = \delta(q_i, a_i, x)$ for some $x\in X$.
A \emph{nominal outcome} of a strategy is a trace $q_0 q_1 q_2\ldots$ where $q_0$ is the initial state and
for each $i\geq 0$ we have $q_{i+1} = \delta(q_i, a_i, \epsilon)$.

A strategy $S$ is {\em memoryless} if $S(w\cdot q) = S(w'\cdot q)$ for all $w,w'\in Q^*$ and $q\in Q$, that is, if
it depends only on the last state on the trace.
In this case, we omit the (irrelevant) prefix, and consider a strategy to be a function from $Q$ to $2^{\Sigma}$.

A strategy $S$ is called \emph{deterministic} if for all $q \in Q$, $|S(q)| = 1$.
If a strategy does not have this property we say that it is \emph{non-deterministic}.

For a state $q \in Q$ of a reachability or B\"uchi automaton with acceptance set $F$ and strategy $S$, let
$\tau^S(q,Q')$ denote the set of nominal traces
connecting $q$ to an element of the set $Q' \subseteq Q$. Note that if $S$ is deterministic the set $\tau^S(q,Q')$ will contain only one trace
for each $q$; we abuse notation and let $\tau^S(q,Q')$ denote this unique trace directly.
Let $Reach_S(q) \subseteq Q$ denote the set of
states in $A$ reachable from $q$ via a finite trace resulting from the system following strategy $S$ and any environmental action, and let $Reach_{ST}(q)$ denote the set of states in $A$ reachable from $q$ via a finite trace resulting from the system following any strategy $S$ and the environment following strategy $T$.

A {\em disturbance strategy} is a function from $Q^+ \times \Sigma$ to $X$.
Let $S: Q^+\to 2^{\Sigma}$ be a strategy and $T: Q^+\times \Sigma \to X$ a disturbance strategy.
An outcome $q_0q_1\ldots$ of $S$ and $T$ is a trace on which $q_0$ is the initial
state of the automaton and for each $i\geq 0$ we have
$q_{i+1} = \delta(q_i, a, T(q_0\ldots q_i, a))$ for $a \in S(q_0\ldots q_i)$.

Let $(A, F)$ be an automaton together with an acceptance condition.
A strategy is {\em nominally winning} in $A$ if every nominal outcome satisfies $F$.
It is known that reachability, B\"uchi, and parity conditions admit memoryless nominally winning strategies \cite{EmersonJutla91}.
A strategy $S$ is {\em winning} if for every disturbance strategy $T$, each outcome of $S$ and $T$
satisfies $F$.

For a finite automaton $A$ and memoryless strategy $S:Q \to 2^{\Sigma}$ let $A |_{S}$ denote the automaton resulting from restricting the behaviour of $A$ using $S$. That is, the automaton $A |_{S}$ has
\begin{itemize}
\item State set $Q^S = Reach_S(q_0) \subseteq Q$ with distance function $d$;
\item Initial state $q_0$;
\item Input alphabet $\Sigma$;
\item Disturbance alphabet $X$;
\item Partial transition function $\delta_S: Q^S \times \Sigma^S \times X \to Q^S$ with $\delta_S(q,a,x) = \delta(q,a,x)$ if $a \in S(q)$, that is, $\delta_S$ is equal to $\delta$ restricted to $Q^S$ by $S$;
\item Disturbance function $\gamma_S: Q^S \to \mathbb R_0^+$ which is the restriction of $\gamma$ to $Q^S$.
\end{itemize}

We are now able to introduce our definitions of robustness.

\begin{definition}\label{sigmarobconst}
A nominally winning strategy $S$ for a reachability or B\"uchi automaton $(A,F)$ is \emph{$\sigma$-robust} if $S$ is winning for the automaton $(A,F')$ where $F' = \{ q \in Q \mid d(q,F) \leq \sigma \overline{\gamma} \}$.

A nominally winning strategy $S$ for a generalized B\"uchi automaton $(A,\mathscr F)$ is \emph{$\sigma$-robust} if $S$ is winning for the automaton $(A,\mathscr F')$ where $\mathscr F' = \{ F_0',\ldots, F_{n-1}'\}$ with
\[ F_j' = \{ q \in Q \mid d(q,F_j) \leq \sigma \overline{\gamma}\}\]
for $j = 0, \ldots, n-1$.

A nominally winning strategy $S$ for a parity automaton $(A,\mathscr F)$ is \emph{$\sigma$-robust} if $S$ is winning for the automaton $(A,\mathscr F')$ where $\mathscr F' =  \{ F'_0, F_1, F'_2, \ldots, F'_{2n}, F_{2n+1}\}$ where
\[ F_{2i}' = \{ q \in Q \mid d(q,F_{2i}) \leq \sigma\overline{\gamma} \}\]
for $i=0,\ldots, n$.
\end{definition}

We show a simple example illustrating our definitions.

\begin{example}\label{runningex}
Consider the reachability automaton $(A,F)$ with $Q = \{q_0,\ldots, q_6\}$, $\Sigma = \{a,b\}$ and $F = \{q_6\}$. The automaton $A$ is equipped with a distance function $d: Q \times Q \to \mathbb R_0^+$. The relative distances of the states in $Q$ are presented in Table \ref{extable} and are approximated by the relative arrangement of the states in the automaton in Figure \ref{exauto}. The disturbance function is defined as $\gamma(q) = 1$ for all $q \in Q$
and the nominal behaviour is defined as shown in Figure \ref{exauto}. Since the disturbance bound is constant we shall refer to it simply as $\overline{\gamma}=1$.

\begin{figure}
\begin{center}
\begin{tikzpicture}[auto,scale=0.9]
\def\stuckellipse{(-3,-0.5) ellipse (1.75cm and 2cm)}
\def\badstratellipse{(1.7,0) ellipse (4.2cm and 2cm)}
\def\goodstratellipse{(4.6,0) ellipse (1.3cm and 1cm)}
\fill[red!50!white] \stuckellipse;
\fill[yellow!50!white] \badstratellipse;
\fill[green!40!white] \goodstratellipse;
\begin{scope}
      \clip \stuckellipse;
      \fill[orange!60!white] \badstratellipse;
    \end{scope}
   \node (q-) at (-3,0) [] {};
  \node (q0) at (-2,0) [draw,fill=white,circle,thick,inner sep=0.7] {$\mathbf{q_0}$};
   \node (q2) at (-4,0.1) [draw,fill=white,circle,thick,inner sep=0.7] {$\mathbf{q_2}$};
   \node (q1) at (-3,-1.1) [draw,fill=white,circle,thick,inner sep=0.7] {$\mathbf{q_1}$};
   \node (q3) at (1.5,0) [draw,fill=white,circle,thick,inner sep=0.7] {$\mathbf{q_3}$};
   \node (q4) at (1.7,1.5) [draw,fill=white,circle,thick,inner sep=0.7] {$\mathbf{q_4}$};
  \node (q5) at (4,0) [draw,fill=white,circle,thick,inner sep=0.7] {$\mathbf{q_5}$};
   \node (q6) at (5.5,0) [draw,circle,fill=white,double,thick,inner sep=0.7] {$\mathbf{q_6}$};
\draw [->,thick] (q-) to (q0);
\draw [->,thick] (q0) to node [swap] {$b$} (q1);
\draw [->,thick] (q0) to node [swap] {$a$} (q3);
\draw [->,thick,bend right=20] (q1) to node [swap] {$a,b$} (q6);
\draw [->,thick] (q2) to node [swap] {$b$} (q1);
\draw [->,thick,bend left=20] (q2) to node {$a$} (q3);
\draw [->,thick] (q3) to node {$a,b$} (q5);
\draw [->,thick,bend left=20] (q4) to node {$a,b$} (q6);
\draw [->,thick] (q5) to node {$a,b$} (q6);
 \end{tikzpicture}
\caption{The undisturbed automaton $A_{\epsilon}$}
\label{exauto}
\end{center}
\end{figure}
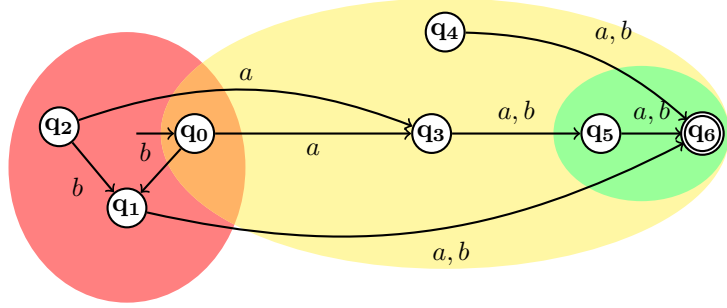

\begin{table}
\begin{center}
\begin{tabular}{c|ccccccc}
\phantom{$q_1$} & $q_0$ & $q_1$ & $q_2$ & $q_3$ & $q_4$ & $q_5$ & $q_6$ \\
\hline
$q_0$ & 0 & 1 & 2 & 4 & 4 & 4 & 5 \\
$q_1$ &   & 0 & 1 & 5 & 5 & 5 & 6 \\
$q_2$ &   &   & 0 & 6 & 6 & 7 & 8 \\
$q_3$ &   &   &   & 0 & 1 & 3 & 3 \\
$q_4$ &   &   &   &   & 0 & 3 & 3 \\
$q_5$ &   &   &   &   &   & 0 & 1 \\
$q_6$ &   &   &   &   &   &   & 0 \\
\end{tabular}
\caption{The distance between states in the automaton $A_{\epsilon}$ of Figure \ref{exauto}.}
\label{extable}
\end{center}
\end{table}

%
Let $S_b:Q \to \Sigma$ be the deterministic memoryless strategy which chooses $b \in \Sigma$ for every $q \in Q$,
and let $S_a:Q \to \Sigma$ be the deterministic memoryless strategy which chooses $a \in \Sigma$ for every $q \in Q$.
Clearly both $S_b$ and $S_a$ are nominally winning for the
reachability condition; they are equally good strategies in classical automata theory.

Consider the result of applying the strategies $S_b$ and $S_a$ in the disturbed automaton $A$.
First note that the unique nominal trace connecting the initial state $q_0$ to the terminal state $q_6$ resulting from applying $S_b$ is $q_0 q_1 q_6$.
Inputting $b$ at state $q_0$ could result in reaching any of the states in the ellipse on the left and hence it is possible that the system may remain in state $q_0$ indefinitely. Then since $q_0$ is at a distance of 5 from the terminal state $q_6$, a trace implementing $S_b$ is only guaranteed to reach a
state at distance $5$ or less
from $F$.
Therefore, in the disturbed automaton, strategy $S_b$ is winning
with respect to the inflated acceptance condition $F_b = \set{q\in Q \mid d(q,F)\leq 5}$ as shown in Figure \ref{exauto} and $S_b$ is 5-robust.

Now consider the strategy $S_a$.
The nominal trace connecting $q_0$ to $q_6$ for this strategy is $q_0 q_3 q_5 q_6$.
Note that $d(q_0,q_3)$ and $d(q_3,q_5)$ are both greater than the power of the disturbance $\overline{\gamma}=1$. Therefore in the disturbed automaton progress is still being made towards $F$ until we reach $q_5$ which is at a distance of 1 from $F$.
Hence the strategy $S_a$ is winning with respect to the inflated reachability condition $F_a = \set{q\in Q \mid
 d(q,F)\leq 1}$ as shown in Figure \ref{exauto} and $S_a$ is 1-robust.

In classical automata and game theory (e.g., \cite{Zielonka98}),
the outcomes of the two strategies are indistinguishable:
both strategies reach the set $F = \{q_6\}$ in the nominal case,
and may result in traces which never reach $F$ when disturbances are present.
However, the metric $d$ provides an extra method of
comparison: the distance from $F$ as a
function of the bound on the disturbance $\overline{\gamma}$.
With this in mind
it is obvious that the strategy $S_a$ is a better choice for the
automaton $A$.

We discuss the construction of the two strategies $S_b$ and $S_a$ in Section \ref{safetysec}.
\end{example}

\section{Reachability}\label{safetysec}

In this section we provide methods to verify the robustness of strategies for finite reachability acceptance conditions, as well as algorithms to synthesize optimally robust strategies. The definitions presented are based upon ideas from continuous control and provide the foundations for dealing with more complex infinite acceptance conditions in the following sections.

Let
$(A,F) = ((Q,d), q_0, \Sigma, \delta, X, \gamma, F)$
be a reachability automaton satisfying Assumption \ref{connectreach}.
A \emph{(reachability) rank function} with respect to $F$ is a function
$R_F:Q \to \mathbb R_0^+$ where $R_F(q) = 0$ if and only if $q \in F$, and
there exists a monotonically increasing function
$\alpha: \mathbb R_0^+ \to \mathbb R_0^+$
satisfying $\alpha(0) = 0$ and
\begin{equation}
\label{LBound}
 \alpha(d(q,F))\leq R_F(q) \quad \mbox{for all }q \in Q.
 \end{equation}
A rank function $R_F$ is said to be a \emph{control Lyapunov function}
if there exists a monotonically increasing function
$f: \mathbb R_0^+ \to \mathbb R_0^+$ satisfying $f(0)=0$ and
such that for each $q \in Q\backslash F$ there exists some $a \in \Sigma$ with
\begin{equation}
\label{Ineq}
R_F(q^{a\epsilon}) - R_F(q) \leq - f(d(q,F)).
\end{equation}
We exclude states in the set $F$ since we exclusively consider finite reachability conditions of the form
$\Diamond F$.
By asking that $R_{F}$ satisfies inequality (\ref{Ineq}) at every state in $Q$ one may also reason about acceptance
conditions of the form $\Diamond \Box F$ (``eventually always $F$'') in the same manner.

%
A control Lyapunov function $R_F$ induces one or more memoryless
strategy functions $S$ defined by mapping a state $q \in Q$ to some subset of the
inputs $a \in \Sigma$ which satisfy inequality~(\ref{Ineq}).

The existence of a control Lyapunov function relies upon the
nominal coreachability assumption with respect to $F$.
This is a natural assumption in certain applications, such
as in the control of physical systems, but it is typically not satisfied in the normal treatment of reachability
of discrete systems. However it is straightforward to restrict the state set of the automaton to exclude states from which the set $F$ cannot be reached via a finite nominal trace.

\begin{theorem}\label{safety}
Let $(A,F)$ be a finite reachability automaton satisfying Assumption \ref{connectreach}.
A memoryless strategy $S$
is nominally winning with respect to $F$
if and only if
there exists a control Lyapunov function $R_F$ such that $S$ can be induced from $R_F$.
\end{theorem}
\begin{proof}
Let $(A,F)$ be a reachability automaton and let $R_F$ be a Lipschitz continuous
control Lyapunov function with respect to $F$.
Define $S:Q \to 2^{\Sigma}$ by
$S(q) = \set{a\in\Sigma \mid a \mbox{ satisfies inequality~\eqref{Ineq}}}$.
Let $\tau=q_0q_1\hdots$ be a nominal outcome of $S$ in $A$.
Since~(\ref{Ineq}) holds for every $q$ on the trace $\tau$ and the function $R_F$ is non-negative,
$R_F$ decreases along $\tau$ and necessarily reaches zero in finitely many steps.
It then follows from~(\ref{LBound}) that $d(q,F)$ is also zero for some $q \in \tau$ appearing
in $\tau$ after a finite prefix since $d(q,F)\le \alpha^{-1}(R_F(q))$ where the inverse
$\alpha^{-1}$ is also a monotonically increasing function vanishing at zero.

Now let $S:Q \to 2^{\Sigma}$ be a nominally winning strategy for $(A,F)$, and let
$\eta:\mathbb R_0^+ \to \mathbb R_0^+$ be a monotonically increasing function.
We define a weighted digraph $G = (Q,E)$ in which there exists an edge
$(q,x_q,q') \in E$ with $x_q = \eta(d(q,F))$ if and only if
$\delta(q,a,\epsilon) = q'$ for some $a \in S(q)$.
For each $q \in Q$ define
\[
R(q) = \min_{\tau \in \tau^S(q,F)}\sum_{q' \in \tau} x_{q'}.
\]
Note that the definition above is indeed well formed: every trace in the set $\tau^S(q,F)$
is simple (that is, one without loops) by definition and therefore each state in $Q$ may appear at most once on such a trace.

Observe that $\eta(d(q,F)) \leq R(q)$ for all $q \in Q$ and hence that $R$ is a reachability rank function.
Also since $R(q) \geq R(\delta(q,a,\epsilon)) + \eta(d(q,F))$ for every $a \in S(q)$ we may trivially observe that
\[ R(\delta(q,a,\epsilon)) - R(q) \leq - \eta(d(q,F)) \]
for every $q \in Q$ and the function $R$ is indeed a control Lyapunov function.

By definition the function $R$ satisfies inequality (\ref{Ineq}) at a state $q$ for an
input $a\in \Sigma$ if and only if $a \in S(q)$, and therefore the strategy induced
from $R$ will be precisely $S$ as required.
\end{proof}

Control Lyapunov functions provide a method for the verification of robustness of potential strategies.
For a system with a constant disturbance bound, the following theorem describes the ``graceful degradation'' or robustness properties possessed by strategies induced from control Lyapunov functions. When disturbances are present, a nominal outcome is not guaranteed but no catastrophic failure will occur.
Instead, the deviation from a nominal outcome is
linearly bounded by the power of the disturbance, and may be explicitly calculated.

\begin{theorem}\label{safetysigma}
Let $(A,F)$ be a finite reachability automaton satisfying Assumption \ref{connectreach} with disturbance bounded by $\overline{\gamma}$ and let $S$ be a
nominally winning memoryless strategy induced from a control Lyapunov function $R_F$.
Then $S$ is a $f^{-1}(K\overline{\gamma})/\overline{\gamma}$-robust winning strategy where $K$ is the Lipschitz constant of $R_F$.
\end{theorem}
\begin{proof}
Assume that $R_F: Q \to \mathbb R_0^+$ is a control Lyapunov function for the reachability automaton $(A,F)$ and let $S:Q \to 2^{\Sigma}$ be a strategy induced from $R_F$.
Let $T$ be a disturbance strategy and consider an outcome $\tau=q_0q_1\hdots$ of $S$ and $T$. We first establish the inequality:
$$R_F(q^{ax}) - R_F(q) \leq K\overline{\gamma} - f(d(q,F))$$
for any $q$ appearing in $\tau$:
\begin{eqnarray}
R_F(q^{ax}) - R_F(q) & = & R_F(q^{ax}) -R_F(q^{a\epsilon}) +R_F(q^{a\epsilon})-R_F(q)\notag\\
& \le& \vert R_F(q^{ax}) -R_F(q^{a\epsilon})\vert -f(d(q,F))\notag\\
& \leq& K d(q^{ax},q^{a\epsilon}) 	-f(d(q,F))\notag\\	
&= & K\overline{\gamma}-f(d(q,F)).\notag
\end{eqnarray}
Note that as long as $q$ is sufficiently far from $F$, the value $-f(d(q,F))$ is sufficiently negative, and the sum $K\overline{\gamma}-f(d(q,F))$ remains negative.
Hence, $R_F$ continues to decrease along $\tau$.
The situation changes when we reach a state $q$ satisfying $K\overline{\gamma}>f(d(q,F))$.
Hence, an outcome of $S$ and $T$ is guaranteed to reach the set
$F' = \{q \in Q \mid f(d(q,F)) \leq K\overline{\gamma}\}$
(or equivalently, $F' = \{q \in Q \mid d(q,F) \leq f^{-1}(K\overline{\gamma})\}$)
in finitely many steps and therefore $S$ is $\sigma$-robust where $\sigma = f^{-1}(K\overline{\gamma})/\overline{\gamma}$.
\end{proof}

The case in which the function $f$ is linear, that is, for every $x \in \mathbb R_0^+$,
$f(x) = cx$ for a fixed constant $c \in \mathbb R_0^+$, is worth noting. In this case the expression $\sigma = f^{-1}(K\overline{\gamma})/\overline{\gamma}$ in the above theorem simplifies to
$\sigma = K/c$.

Although a similar approach to the one provided in the proof of Theorem \ref{safetysigma} for calculating robustness bounds may be used for automata
having state dependent disturbance bounds, the resulting value is likely to be conservative.
Indeed, let $(A,F)$ be a finite reachability automaton and let $S$ be
a memoryless strategy with associated control Lyapunov function $R$.
Let
\[
Q' = \{ q \in Reach_S(q_0) \mid \exists p \in Reach_S(q_0),\exists a \in S(p) : \delta(p,a,\epsilon)=q \wedge K\gamma(q) > f(d(p,F)), \}
\]
 the set of states where the control Lyapunov inequality (\ref{Ineq}) may be violated under the effects of a disturbance, that is, states $q$ from which the disturbance action can force the system to reach a state which is further away from the target set than $q$.
The value of $\sigma$ calculated via the method presented in Theorem \ref{safetysigma} would be
\[ \sigma = \frac{\max \{ d(q,F) \mid q \in Q' \}}{\overline{\gamma}}.\]
Let $q \in Q$ be the state achieving this value, that is, $d(q,F) = \sigma \overline{\gamma}$ and let $p \in Q$ be the state reached by following $S$ at $q$. If
$\delta(p,a,x) \not\in Q'$
for all $x \in X \setminus \{\epsilon\}$,
a smaller value of the bound could be achieved.

Instead, for systems with state dependent disturbance bounds, we give
a dynamic programming algorithm. 
The operators presented below will form the basis for optimal synthesis and robustness verification not
just for reachability automata, but for the $\omega$-regular automata which follow in later sections.

Fix a reachability automaton $(A,F)$, and let $Q = \set{q_0,\ldots, q_{m-1}}$.
We characterize the optimal robustness bound achievable by a memoryless strategy 
as the fixed point of a certain operator.
The operator acts upon a vector of size $|Q|=m$ consisting of positive real numbers.

Consider a state $q\in Q$ and the objective to reach the set $F$ via a finite trace beginning at $q$.
We argue that for any nominally winning strategy $S$ beginning at $q$, the robustness bound
$\sigma$ cannot be more than $d(q,F)/\overline{\gamma}$, since just by staying at $q$, the strategy
ensures that the system is within distance $d(q,F)$ of the final states
(c.f. strategy $S_b$ in Example \ref{runningex}).
Hence the maximal value of $\sigma$ is equal to $d(q,F)/\overline{\gamma}$.

We define a sequence of vectors $\opt^i$ for $i\geq 0$.
With the above intuition, we define
$\opt^0(j) = d(q_j,F)$ for $j=0,\ldots,m-1$.

For $q \in Q$ let
$\Post_a(q) = \{ q' \mid \exists x\in X \ \delta(q,a,x) = q' \} \subseteq Q$,
the set of states reachable from $q$ via the input action $a$.
The definition is extended to sets of states in the natural way.
Further, for words $w = w_1 \ldots w_n \in \Sigma^*$, we write
\[ \Post_{w}(q) = \Post_{w_n}(\Post_{w_{n-1}}( \ldots \Post_{w_1}(q))), \]
with the assumption that $\Post_\lambda(q) = q$ for the empty word $\lambda$.

\begin{definition}\label{itdef}
Define the monotonic operator $\underline{g}: (\mathbb R_0^+)^{m} \to (\mathbb R_0^+)^{m}$ by
\[
\underline{g}(\opt)(j) = \min \left( opt(j),\min_{a\in \Sigma} \left(\max_{q_i \in Post_a(q_j)} opt(i) \right)\right).
\]
Let $\opt^{i+1} = \underline{g}(\opt^i)$.
\end{definition}

Consider the result of applying $\underline{g}$ once to the vector $opt^0$.
As previously stated, $opt^0(j) = d(q_j,F)$ for each $1 \leq j \leq m$. Applying $\underline{g}$ gives the result
\begin{align*}
opt^1(j) =& \min \left( opt^0(j), \min_{a\in \Sigma} \left( \max_{q_i \in Post_a(q_j)} opt^0(i) \right)\right) \\
         =& \min \left( d(q_j,F), \min_{a\in \Sigma} \left( \max_{q_i \in Post_a(q_j)} d(q_i,F) \right)\right) \\
         =& \min_{a\in \Sigma \cup \{\lambda\}} \max_{q_i \in Post_a(q_j)} d(q_i,F),
\end{align*}
where $\lambda$ is the empty word.
So $opt^1(j)$ encodes the closest the system is able to get to $F$ via a trace of length at most one beginning
at $q_j$ when the environment chooses disturbance inputs which are worst-case,
that is, the environment's objective is to force the system to move as far away as possible from $F$.
Iterating this reasoning leads to the fixed point $\underline{\opt}^*$ defined by
\[
\underline{\opt}^*(j) = \min_{w \in \Sigma^*} \max_{q_i \in Post_w(q_j)} d(q_i,F).
\]
Since the automaton $A$ is finite 
the fixed point $\underline{\opt}^*$ will be reached in a finite number of iterations.
There can be at most $|Q|-1$ iterations since this is the longest input word labeling a simple path between two states in $A$, and
each iteration can be performed in time polynomial in the size of $Q$.
Hence the overall worst case complexity for the algorithm is polynomial in the size of $Q$.

This algorithm is easily seen to be a simple generalization of the Bellman-Ford shortest path algorithm \cite{Bellman54},
modified to take into account the non-determinism resulting from disturbances.

We first use Definition \ref{itdef} 
to verify robustness for a given strategy.
Given a nominally winning memoryless strategy $S$ for a finite
reachability automaton $(A,F)$ the robustness bound $\sigma$ for $S$ is precisely
\[ \sigma = \frac{\underline{opt}^*(0)}{\overline{\gamma}}
\]
for the automaton $A|_{S}$ where $q_0$ is the initial state.

Finally we approach the issue of the synthesis of optimally robust winning strategies.
Given a finite reachability automaton $(A,F)$ the optimal achievable robustness bound for $A$ is
\[
\sigma_{min}= \frac{\underline{\opt}^*(0)}{\overline{\gamma}}
\]
A memoryless strategy achieving the optimal robustness for $(A,F)$ may be recovered in the following way.
We define
$S(q) = \set{a\in\Sigma\mid q_j = \delta(q,a,\epsilon) \mbox{ and } \underline{\opt}^*(j) = \underline{\opt}^*(0)}
        \setminus
        \set{a\in \Sigma \mid q \in \Post_a(q)}$ if the right-hand side is non-empty,
and $S(q) = \Sigma$ otherwise.

\begin{example}
Returning to Example \ref{runningex}, we discuss the two rank functions $R_b: Q \to \mathbb R_0^+$ and $R_a: Q \to \mathbb R_0^+$ from which the strategies
 $S_b$
and $S_a$ are induced. Table \ref{exvaluetable} lists the distance from each state to the terminal state $q_6$ and the value of the two rank functions $R_b$
and $R_a$.

The function $R_b:Q \to \mathbb R_0^+$ is the result of a classical graph theoretic shortest path approach - each
state $q \in Q$ is mapped to the length of the shortest path connecting $q$ to some state in $F$.

Let $\eta: \mathbb R_0^+ \to \mathbb R_0^+$ be the monotonically increasing function defined by $x \mapsto 2x$ for all $x \in \mathbb  R_0^+$.
Then $R_a$ is a control Lyapunov function since for all $q \in Q \setminus \{q_6\}$
\[
R_a(q^{a\epsilon}) - R_a(q) \leq -\eta(d(q,q_6)).
\]
For optimal robustness, the vectors $\opt^0$ and $\underline{\opt}^*$
are as follows for this example.
\[
\opt^0 = \left[ \begin{array}{c}5 \\ 6 \\ 8 \\ 3 \\ 3 \\ 1 \\ 0 \end{array} \right], \quad
\underline{\opt}^* = \left[ \begin{array}{c} 1 \\ 1 \\ 1 \\ 1 \\ 1 \\ 1 \\ 1 \end{array} \right],
\]
Therefore the strategy $S_a$ is optimal with respect to a disturbance of size $\gamma=1$.

\begin{table}[t]
\begin{center}
\begin{tabular}{c|ccc}
$q$ & $d(q,q_6)$ & $R_b(q)$ & $R_a(q)$ \\
\hline
$q_0$ & 5 & 2 & 18 \\
$q_1$ & 6 & 1 & 12 \\
$q_2$ & 8 & 2 & 24 \\
$q_3$ & 3 & 2 & 8  \\
$q_4$ & 3 & 1 & 6  \\
$q_5$ & 1 & 1 & 1  \\
$q_6$ & 0 & 0 & 0  \\
\end{tabular}
\caption{The rank functions $R_b$ and $R_a$ for the automaton $A$.}
\label{exvaluetable}
\end{center}
\end{table}
\end{example}

\smallskip
\noindent{\bf Comparison with existing work.}
At this point it is convenient to compare our framework with the frameworks of Bloem et. al.~\cite{Bloem09} and of Tarraf et. al.~\cite{Tarraf08}.
Both of these references adopt an input-output perspective by relating environment errors (inputs) to system errors (outputs).
In contrast, we adopt a state-space approach by endowing the set of states with a metric and placing no assumptions on the environment other than having bounded power.

In \cite{Bloem09} the authors define the notion of \emph{$k$-robustness} for automata. For a reachability automaton $(A,F)$ two monotonically increasing functions which map zero to zero are defined: an \emph{environmental error function} $e:\Sigma^* \to \mathbb N \cup \{\infty\}$ and a \emph{system error function} $s: \Sigma^* \to  \mathbb N \cup \{\infty\}$. A pair $(e,s)$ of error functions for a given automaton is called an \emph{error specification} for $A$. Then a strategy $S:Q \to \Sigma$ for $A$ is $\kappa$-robust with
respect to the error specification $(e, s)$ if there exists $\beta \in \mathbb N$ such that for all $w \in \Sigma^*$ which label outcomes of $S$,
\[ s(w) \leq \kappa e(w) + \beta.\]

In order to compare Bloem's and Tarraf's results with ours, we resort to some key ideas from robust control~\cite{L2Gain,ROBUSTCONTROL_TEXTBOOK}. First, we define an environment error signal $\mathsf{e}=\mathsf{e}_1\mathsf{e}_2\hdots \mathsf{e}_n\in\R^*$ and a system error signal $\mathsf{s}=\mathsf{s}_1\mathsf{s}_2\hdots \mathsf{s}_n\in \R^*$. The only assumption we place on $\mathsf{e}$ and $\mathsf{s}$ is that an absence of environment errors at time $k\in \N$ corresponds to $\mathsf{e}_k=0$ and the absence of system errors at time $k\in \N$ corresponds to $\mathsf{s}_k=0$. The error functions $e$ and $s$ in Bloem's framework can be seen as the cumulative versions of $\mathsf{e}$ and $\mathsf{s}$, for example:
$$e(k)=\sum_{i=0}^{k}\mathsf{e}_i,\qquad s(k)=\sum_{i=0}^{k}\mathsf{s}_i.$$
In Tarraf's framework and notation the role of $\mathsf{e}_i$ is played by $\rho(u(i))$ and the role of $\mathsf{s}_i$ is played by $\mu(y(i))$.
We now regard an automaton as defining a transformation $F:\R^*\to \R^*$ from environment error signals to system error signals $F(\mathsf{e})=\mathsf{s}$. In general
$F$ will be a set valued function, but we assume it to be single valued to simplify the discussion.


The notion of finite-gain stability from robust control can now be introduced as follows:

A map $F:\R^*\to\R^*$ is said to be \emph{finite-gain stable} with \emph{gain} $\kappa$ and \emph{bias} $\beta$ if the following inequality holds:
\begin{equation} \label{FiniteGain}
\sum_{i=0}^{n}F(\mathsf{e})\le \kappa\sum_{i=0}^{n}\mathsf{e}+\beta
\end{equation}
for every $\mathsf{e}\in \R^*$. A more condensed version of~(\ref{FiniteGain}) is:
$$s\le \kappa e+\beta$$
which is Bloem's notion of $\kappa$-robustness and Tarraf's notion of $\rho/\mu$ gain stability.
It is well known in robust control and dissipative systems theory that the existence of
a certain type of Lyapunov function (a storage function) implies finite-gain stability.
In the context of reachability automata, we define $\mathsf{e}$ to be the effect of
the environment actions on the state: \[\mathsf{e}_k=d(q_k^{ax},q_k^{a\epsilon}).\]

If $x=\epsilon$ then $q_k^{ax} = q_k^{\epsilon}$ and $\mathsf{e}_k=0$, since the behaviour coincides with the nominal behaviour under no environment disturbances. For problems of the form $\Diamond \Box F$ we regard $F$ as the set of states describing the desired operation for the system. Hence, any deviation from $F$ is regarded as a system error.
The system error signal is defined as:

$$\mathsf{s}_k=d(q_k,F).$$

Standard arguments in dissipative systems theory~\cite{L2Gain} would then show that:
$$s\le f^{-1}(K e)+R_F(q_0)$$
where  $f:\mathbb R_0^+ \to \mathbb R_0^+$ is some monotonically increasing function satisfying $f(0)=0$ and $K$ is the Lipschitz constant of $R_F$.
It is also known that finite-gain stability does not imply the notion of stability
considered in this paper unless certain controllability/observability assumptions hold.
This follows from the fact that it may not be possible to infer the decrease of $R_F$
at every state only from the knowledge of $\mathsf{e}$ and $\mathsf{s}$
when not every state can be reached from $q_0$ or when $\mathsf{s}$ does not provide enough information about the state.

\section{Example}

\begin{figure}
\begin{center}
\begin{tikzpicture}[scale=0.6]
\node (1) at (0,2) [draw,circle] {1};
\node (2) at (-2,0) [draw,circle] {2};
\node (3) at (2,0) [draw,circle] {3};
\node (4) at (0,-2) [draw,circle] {4};

\draw (1) to (2);
\draw (2) to (4);
\draw (1) to (3);
\draw (3) to (4);
\end{tikzpicture}
\caption{The communication network for leadership election.}
\label{netex}
\end{center}
\end{figure}
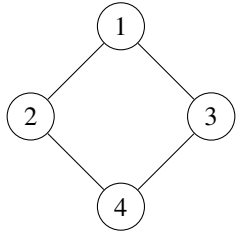

In this section we recast a classic problem from distributed computing in our framework to allow
the explicit quantification of the robustness of possible solution strategies.
Figure \ref{netex} shows a network of four computer nodes, each having a two way
communication channel (represented by an undirected edge in the graph)
connected to each of its two neighbouring nodes.
Each computer in the network has a unique identifier which is presented in the figure.
The four nodes are required to elect a leader, and may make use of the communication channels to exchange information.
In order for a leader to be elected, the nodes must come to a unanimous consensus on which of the four nodes is the leader.
However, the communication channels between the nodes are known to be subject to noise,
and so messages may be corrupted between transmission and receipt, as described below.

We model the system as an automaton with state set defined by the global state of the network.
That is, each state in the automaton represents the current belief of the four nodes as to who is the leader.
Hence $Q \subseteq \{1,2,3,4\}^4$.
The initial state is $(1,2,3,4)$.
At each state, each of the four nodes communicates its current belief to its neighbouring nodes,
and each node uses this information combined with its own belief about the current leader
to update its belief.
The acceptance condition is a reachability condition with terminal set $\set{(i,i,i,i)\mid i\in\set{1,2,3,4}}$.
There are a number of different strategies which the nodes could apply to decide upon a new belief
using the information available to them.
We consider the following three possibilities.
\begin{description}
\item[B] Each node chooses the least of the three values;
\item[T] Each node chooses the largest of the three values;
\item[F] Each node chooses the integer part of the average of the three values.
\end{description}
It is a well known result in distributed computing that choosing either of the first two strategies is computationally optimal
\cite{Lynch}.

The disturbances in our system are characterized in the following way: beliefs are assumed to be sent as decimal numbers,
and the noise in the channel may cause the value of the sent belief to change by $\pm 1$.
However, we do not allow messages outside of the set $\{1,2,3,4\}$: for
example if a disturbance occurs on the message `1', the recipient will receive either `1' or `2'.
A distance function on the state set $Q$ is defined by
\[ d((x_1x_2,x_3,x_4),(y_1,y_2,y_3,y_4)) = |x_1-y_1| + |x_2-y_2| + |x_3-y_3| + |x_4-y_4|;\]
this is precisely the \emph{Manhattan} or $L_1$ norm.
For each node $i \in \{1,2,3,4\}$ we assume that only one of the two incoming messages may be affected
by the disturbance at any given time in order to simplify the presentation, though the methodology applies in the same way without
this assumption.
This combined with our assumption about the power of the disturbance on the messages themselves translates into a
constant disturbance in our automaton model of size $\overline{\gamma} =1$.


\begin{figure}
\centering
\subfloat{Strategy $B$}{\label{fourbottom}\includegraphics[width=5cm]{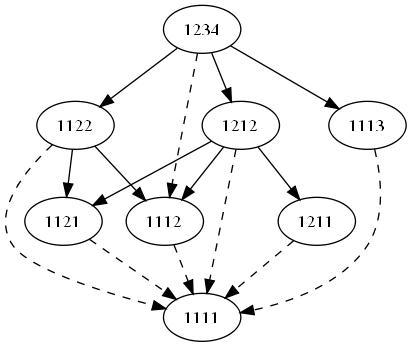}}
\subfloat{Strategy $T$}{\label{fourtop}\includegraphics[width=5cm]{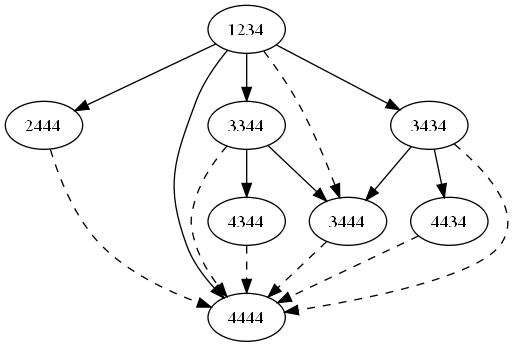}}
\caption{Classical optimal strategies for the leader election problem.}
\label{fouropt}
\end{figure}

%


  \begin{figure*}
  \begin{center}
   {\includegraphics[width=8cm]{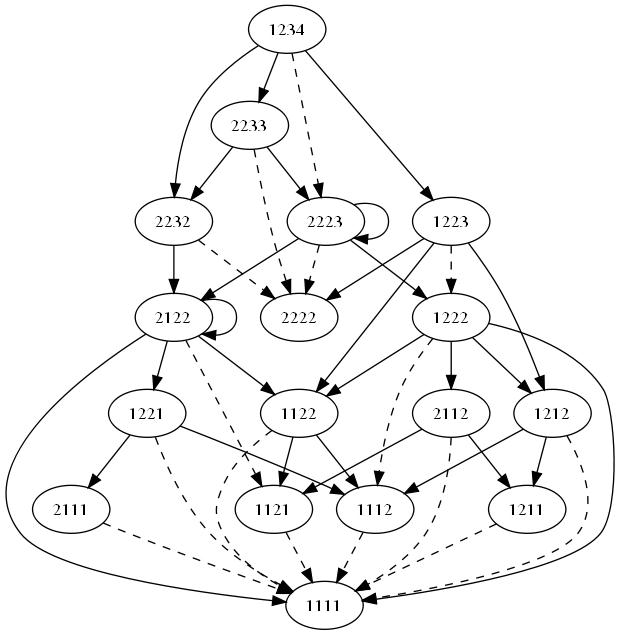}}
   \caption{Leader election using strategy $F$.}
   \label{fourave}
   \end{center}
   \end{figure*}

Figures \ref{fouropt} and \ref{fourave} show the metric automata for the three strategies described above.
We restrict to the reachable part of the automaton.
Nominal transitions are represented by dashed lines, disturbed transitions by solid.

We use the function $\underline{opt}$ described in Section \ref{safetysec} to analyze the robustness of the three strategies.
First note that strategies $B$ and $T$, the classical optimal strategies, are 0-robust.
Indeed, the fixed point iteration gives $\sigma_{\min}= 0$.
More interesting is the conclusions we may draw for the floor strategy $F$.
Here $\sigma_{\min} = 1$, due to the self loops at states
$2223$ and $2232$.
Hence a disturbance bounded by $\gamma=1$ results in only one node having the wrong belief and hence
though the nodes do not reach a unanimous decision, they at least are able to come to a majority decision.
This is obviously a ``better'' outcome than that resulting from only two nodes agreeing on their belief.

\section{Omega-regular objectives}

We now extend the results to more general $\omega$-regular acceptance conditions.
We do this in two steps.
First, we provide a simple generalization to B\"uchi acceptance
conditions.
Then, we show how ideas based on \emph{progress measures}
\cite{klarlund90,Namjoshi01}
can be used to provide robustness results for parity acceptance conditions.
In every case we make an appropriate connectedness assumption.

\subsection{B\"uchi acceptance conditions}\label{buchisec}

Let $(A,F) = ((Q, d), q_0, \Sigma, \delta, X, \gamma, F)$ be a B\"uchi automaton
with acceptance condition $F \subseteq Q$ such that
$A$ is nominally coreachable with respect to $F$.
First note that B\"uchi acceptance asks that for a trace
$\tau \in Q^{\omega}$, the intersection
of the set $\zeta(\tau)$ with the set of terminal states is non-empty.
So by viewing the B\"uchi condition as an infinite series of
reachability conditions for the set $F$, and under
Assumption \ref{connectreach}, the definitions and results
for reachability also apply in the case of B\"uchi automata.

In particular, note that the definition of a control Lyapunov function given
in the previous section only requires that inequality (\ref{Ineq}) holds
for states outside of the set $F$.
A control Lyapunov function $R_F$ for a B\"uchi automaton induces a memoryless
strategy $S: Q\to 2^{\Sigma}$ which specifies
actions satisfying \eqref{Ineq} for any state in $Q\setminus F$ and any
arbitrary action for states in $F$.
The strategy $S$ is nominally winning: the argument that $F$ is reached is identical
to the reachability case, and the coreachability assumption
ensures that an arbitrary action from $F$ will not prevent $F$ from being visited again.


\begin{proposition}\label{buchi}
Let $(A,F)$ be a finite B\"uchi automaton satisfying Assumption \ref{connectreach} and let $S$ be a memoryless strategy. Then $S$ is nominally winning if and only if there exists a Lipschitz continuous control Lyapunov function $R_{F}$ such that $S$ can be induced from $R_{F}$.
\end{proposition}

Since we are able to cast B\"uchi acceptance as an infinitely repeated reachability condition, the methods for calculating $\sigma$ for a given strategy and optimal achievable robustness bounds are identical to the reachability case.

\begin{proposition}\label{buchisigma}
Let $(A,F)$ be a finite B\"uchi automaton satisfying Assumption \ref{connectreach} with constant disturbance bound and let $S$ be a nominally winning memoryless strategy induced from a control Lyapunov function $R_F$.

$S$ is a $f^{-1}(K\overline{\gamma})/\overline{\gamma}$-robust winning strategy where $K$ is the Lipschitz constant of $R_F$.
\end{proposition}

The robustness bound for a given nominally winning strategy may be calculated in a manner identical to that presented for reachability automata. The same is true for the optimal and worst case achievable robustness bounds and optimal strategies for a given B\"uchi automaton.

%


\begin{example}\label{buchiex}
Consider the B\"uchi automaton $(A,F)$ with $F = \{q_6\}$ whose nominal behaviour is shown in Figure \ref{exautobuchi}.
Note that this automaton is identical to the reachability
automaton presented in Figure \ref{exauto} (Example \ref{runningex}) with the addition of two new edges beginning at $q_6$.
The distances between the states and the rank functions $R_b:Q \to \mathbb R_0^+$
and $R_a:Q \to \mathbb R_0^+$ are as before; their values
may be found in Tables \ref{extable} and \ref{exvaluetable}.
The two strategies $S_b: Q \to \Sigma$ and $S_a: Q\to \Sigma$ are induced in the same way for states in $Q \setminus F$.
Observe that a control Lyapunov function for a B\"uchi automaton does not specify the value of the induced strategy for terminal states. There are of course two options, namely $a$ and $b$, leading to the states $q_0$ and $q_2$ respectively.

For $S_a$ observe that $R_a(q_0) < R_a(q_2)$ and so we set $S_a(q_6)$ equal to $a$.
For the strategy $S_b$, note that $R_b(q_0) = R_b(q_2)$. For consistency we set $S_b(q_6) = b$. Then the strategy $S_b$ is $5$-robust and $S_a$ is $1$-robust.

\begin{figure}
\begin{center}
 \begin{tikzpicture}[auto]
\node (q-) at (-3,0) [] {};
  \node (q0) at (-2,0) [draw,fill=white,circle,thick,inner sep=0.7] {$\mathbf{q_0}$};
   \node (q2) at (-4,0.1) [draw,fill=white,circle,thick,inner sep=0.7] {$\mathbf{q_2}$};
   \node (q1) at (-3,-1.1) [draw,fill=white,circle,thick,inner sep=0.7] {$\mathbf{q_1}$};
   \node (q3) at (1.5,0) [draw,fill=white,circle,thick,inner sep=0.7] {$\mathbf{q_3}$};
   \node (q4) at (1.7,1.5) [draw,fill=white,circle,thick,inner sep=0.7] {$\mathbf{q_4}$};
  \node (q5) at (4,0) [draw,fill=white,circle,thick,inner sep=0.7] {$\mathbf{q_5}$};
   \node (q6) at (5.5,0) [draw,circle,fill=white,double,thick,inner sep=0.7] {$\mathbf{q_6}$};
\draw [->,thick] (q-) to (q0);
\draw [->,thick] (q0) to node [swap] {$b$} (q1);
\draw [->,thick] (q0) to node [swap] {$a$} (q3);
\draw [->,thick,bend right=20] (q1) to node [swap] {$a,b$} (q6);
\draw [->,thick] (q2) to node [swap] {$b$} (q1);
\draw [->,thick,bend left=20] (q2) to node [pos=0.4] {$a$} (q3);
\draw [->,thick] (q3) to node {$a,b$} (q5);
\draw [->,thick,bend left=20] (q4) to node [swap] {$a,b$} (q6);
\draw [->,thick] (q5) to node {$a,b$} (q6);
\draw [->,thick,bend right=53] (q6) to node [swap] {$a$} (q0);
\draw [->, thick,bend right=60] (q6) to node [swap] {$b$} (q2);
 \end{tikzpicture}
\end{center}
\caption{The undisturbed B\"uchi automaton $A_{\epsilon}$}
\label{exautobuchi}
\end{figure}
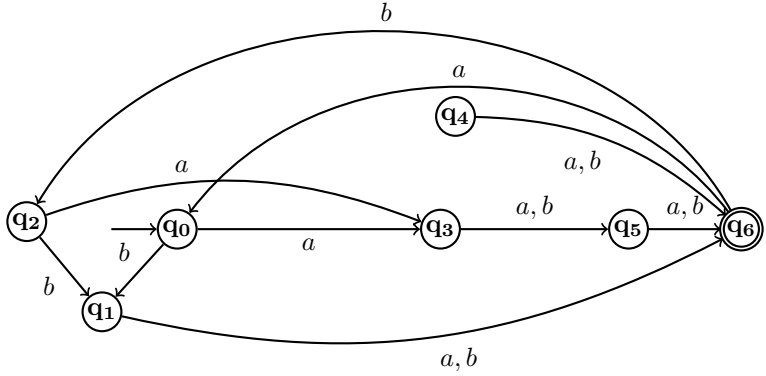
\end{example}

\subsection{Generalized B\"uchi conditions}\label{subgenbu}

We want to generalize the construction of rank functions to parity acceptance conditions.
As a warm-up, we first describe methods for generalized B\"uchi acceptance conditions.
It is a standard argument in automata theory to reduce a generalized B\"uchi automaton to a B\"uchi automaton: the resulting
automaton will have state set $Q \times \{0, \ldots, n-1\}$ where $|\mathscr F|=n$.
So for a system presented as a generalized B\"uchi automaton,
Proposition~\ref{buchi} may be applied to an expanded state space, and winning strategies may be induced.
However, we give an alternate ``direct'' rank function construction based on progress measures that will introduce techniques
useful in the parity case.
Calculating robustness directly for generalized B\"uchi automata has other advantages too: for example, given a distance function $d$ on a generalized B\"uchi automaton $A$, how do we lift $d$ to a metric on the new B\"uchi automaton that makes sense in the context of the original system? This question is likely to be difficult to answer in a satisfactory manner.

Let  $(A,\mathscr F) = ((Q,d), q_0, \Sigma, \delta, X,\gamma, \mathscr F)$ be a generalized B\"uchi automaton with $\mathscr F = \{F_0, \ldots, F_{n-1}\}$.
For \mbox{$i =  0, 1, \ldots, n-1$} let $R_i: Q \to \mathbb R_0^+$ be a (reachability) rank function with respect to the set $F_i$.
Then a \emph{(generalized B\"uchi) rank function} $R: Q \to (\mathbb R_0^+)^n$ is defined by $R(q) =(R_0(q),R_1(q),\ldots,R_{n-1}(q))$ for each $q \in Q$.

We extend the notion of Lipschitz continuity for functions in the obvious way:
a function \mbox{$R:Q \to (\mathbb R_0^+)^{n}$} is \emph{Lipschitz continuous} if there exists $K>0$ such
that for each $i \in \{0,\ldots, n-1\}$
and for all $q,q' \in Q$ it holds that  \[|R_i(q) - R_i(q')| \leq Kd(q,q').\]

As before, if the set $Q$ is finite then every real valued function of this form has this property.

A relation and ordering on $n$-tuples of positive reals is defined as follows.
For every \mbox{$i\in \{0,1,\hdots,n-1\}$} define the preorder $>^i$ on $(\R_0^+)^n$: let $a, b \in (\mathbb R_0^+)^n$ with $a = (a_0, \ldots, a_{n-1})$ and $b = (b_0,\ldots, b_{n-1})$.
Then $a >^i b$ if and only if $a_i > b_i$. We also let $a \geq^i b$ if and only if $a_i \geq b_i$.
Based on $>^i$ we introduce another relation on $(\R_0^+)^n$, denoted by $\RHD^i$ and defined by $a \RHD^i b$ if and only if one of the following two conditions holds:

\[  \ a >^i b \quad \textrm{or} \quad  \ a_{(i-1)\mod n} = 0. \]

Observe that, since the labeling of the sets in $F$ begins at $0$ instead of $1$, the relation $>^0$ corresponds
with the 1st index of the $n$-tuple, $>^1$ corresponds
with the 2nd index, and so on.


\begin{proposition}\label{genbuchirel}
Let $A$ be a finite generalized B\"uchi automaton with acceptance condition $\mathscr F = \{F_0, \ldots, F_{n-1}\}$.
If a trace $\tau = q_0 q_1 q_2 \ldots$  is such that
\begin{align}\label{RHD}
R(q_0) \RHD^0 R(q_1) \RHD^0 R(q_2) \ldots R(q_{i_0}) \RHD^1 R(q_{i_0 +1}) \RHD^1  \ldots \notag\\ \ldots \RHD^1 R(q_{i_{1}}) \RHD^2 R(q_{i_{1}+1}) \RHD^2 \ldots \notag\\
\ldots \RHD^{n-1}  R(q_{i_{n-1}}) \RHD^{0} R(q_{i_{n-1}+1})  \ldots
\end{align}
then $\tau$ satisfies the generalized B\"uchi acceptance condition $\mathscr F$.
\end{proposition}
\begin{proof}
Let $\tau$ be a trace of the form given above. By definition, if two consecutive relations in (\ref{RHD}) have different indices (say $k$ and $k+1$) then the state appearing between them must be contained in the set $F_{k}$. Hence
\[ q_{i_0} \in F_0, q_{i_1} \in F_1, \ldots , q_{i_{n-1}} \in F_{n-1}, \ldots,\] and $\tau$ infinitely often features a state in each of the sets in $\mathscr F$.
\end{proof}

Intuitively, a trace of this form is initially moving towards the set $F_0$ via the relation $\RHD^0$. Once a state in the set $F_0$ is reached,
the second part of the definition of $\RHD$ applies and $\RHD^1$ is satisfied until a state in the set $F_1$ is reached. On reaching a state in the set $F_{n-1}$,
the relation returns to $\RHD^0$, and so on.

Note that the other direction does not necessarily hold: a winning trace will not necessarily have the above form.
For example, the trace may visit the sets in a non-sequential order, or may visit multiple states from each set on each pass through the automaton.

For brevity, we introduce some more notation.
Let $d(q,\mathscr F)$ denote the vector valued distance
\[
d(q,\mathscr F) = ( d(q, F_0), d(q, F_1), d(q,F_2), \ldots, d(q, F_{n-1})).
\]
A generalized B\"uchi rank function $R$ is said to be a \emph{control Lyapunov function}
if there exists a monotonically increasing function $f: \mathbb R^+_0 \to \mathbb R_0^+$
with $f(0) = 0$ such that
for every $i \in \{0, 1, \ldots, n-1\}$ and every $q \in Q\setminus F_i$ there exists $a \in \Sigma$ with
\[ R(q^{a\epsilon}) - R(q) \leq^i - f(d(q, \mathscr F)). \]

For a fixed $i$, the function $R_i$ is a reachability control Lyapunov function with respect to the set $F_i$. Hence every state $q \in Q$ is coreachable with respect to the set $F_i$ for every $i \in \{0,\ldots, n-1\}$ and the automaton $A$ satisfies Assumption \ref{connectgenbuchi}. To see that this is necessary, consider for example a state $q \in Q$
from which the set $F_{i}$ is not reachable for some $i>0$. Then any state coreachable with respect to $q$, and any state reachable from $q$, may not
appear on a winning trace. Hence all such states are redundant (including $q$). These definitions are the natural extension of those given for reachability and B\"uchi automata.

Generalized B\"uchi automata do not admit memoryless strategies; a winning strategy must keep track of the index $i+1$ where $i$ is the index of the last terminal set $F_i$
which was visited on the trace.
Therefore a strategy for a generalized B\"uchi automaton $(A,\mathscr F)$
is a function \mbox{$S: Q \times \{0,\ldots, n-1\} \to 2^{\Sigma}$} where for every
$i\in \{0,\hdots,n-1\}$, the restriction $S(\cdot,i)$ is a memoryless reachability strategy, and may be induced from $R_i$.

\begin{proposition}\label{genbuchi}
Let $(A,\mathscr F)$ be a finite generalized B\"uchi automaton satisfying Assumption \ref{connectgenbuchi} and let \mbox{$S:Q \times \{0,\ldots,n-1\} \to 2^{\Sigma}$} be a memoryless strategy.
Then $S$ is nominally winning if and only if there exist Lipschitz continuous rank functions
$R_i: Q \to \mathbb R_0^+$ for $i =0, 1, \ldots, n-1$
 and a control Lyapunov function $R=(R_0,\hdots,R_{n-1})$ for $(A,\mathscr F)$ such that $S$ may be induced from $R$.
\end{proposition}
\begin{proof}
Straightforward generalization of Proposition \ref{buchi}.
\end{proof}

For automata with constant disturbance bounds we have the following.

\begin{proposition}\label{genbuchisigma}
Let $(A,\mathscr F)$ be a finite generalized B\"uchi automaton satisfying Assumption \ref{connectgenbuchi} with constant disturbance bound $\overline{\gamma}$ and let $S$ be a nominally winning memoryless strategy induced from a control Lyapunov function $R_{\mathscr F}$.
The strategy $S$ is $\sigma$-robust where
\[ \sigma = f^{-1}(K\overline{\gamma})/\overline{\gamma}\]
 for $K = \max_{i=0,\ldots,n-1} K_i$ where $K_i$ is the Lipschitz constant of the rank function $R_i$.
\end{proposition}
\begin{proof}
Assume that $R_{\mathscr F}$ is a control Lyapunov function for $(A,\mathscr F)$ and let $S$ be a nominally winning memoryless strategy induced from $R_{\mathscr F}$. Let $T$ be a disturbance strategy.
Proposition \ref{buchisigma} implies that
\[ R(q^{ax}) - R(q) \leq^i K_i\overline{\gamma} - f(d(q, \mathscr F))\]
for every $q \in \tau \setminus F_i$ where $\tau \in Q^{\omega}$ is any outcome resulting from $S$ and $T$ and with
$K_i$ the Lipschitz constant of $R_i$ with respect to $d$. Hence $S$ is
$\sigma_i$-robust for $\sigma_i = f^{-1}(K_i\overline{\gamma})/\overline{\gamma}$ with respect to $F_i$ and therefore the robustness of $S$ is certainly bounded by $\sigma$ as required.
\end{proof}

For generalized B\"uchi automata with state dependent disturbance bounds the verification of robustness
for a strategy and the calculation of optimal robustness bounds is done in a similar manner to the reachability case.
Let $(A, \mathscr F)$ be a generalized B\"uchi automaton, and assume $Q= \set{q_0,\ldots,q_{m-1}}$ and $|\mathscr F| = n$.
Instead of a vector, we define $\opt^0$ to be an $m$ by $n$ matrix.
Letting $\opt^0(j,k)$ denote the entry in the $j$th row and $k$th column of $\opt^0$, we let $\opt^0(j,k) = d(q_j,F_{k-1})$ for $j=1,\ldots,m$ and $k=1,\ldots,n$.
This is the natural generalization of the definition for reachability and B\"uchi conditions where only one
terminal set is considered.
Then the monotonic function $\underline{g}: (\mathbb R_0^+)^{m \times n} \to (\mathbb R_0^+)^{m \times n}$
is defined on each index $(j,k)$ of the matrix $\opt^0$ by
\[
\underline{g}(\opt)(j,k) = \min \left( \opt(j,k),\min_{a\in \Sigma} \max_{i \in \Post_a(q_j)} \opt(i,k) \right).
\]
That the operator repeatedly applied beginning with $\opt^0$ converges to the required value follows easily from the reachability case.

Given a nominally winning strategy $S$ for a finite generalized B\"uchi automaton $(A,\mathscr F)$ the
robustness bound $\sigma$ may be recovered by first calculating $\underline{opt}^*$ for the restricted automaton $A|_S$. Then
\[
\sigma= \frac{\max_{k = 1, \ldots, n}\underline{opt}^*(0,k)}{\overline{\gamma}}.
\]
For optimal strategy synthesis we calculate the minimal achievable robustness bound as
\[
\sigma_{\min} = \frac{\max_{k = 1, \ldots, n}\underline{opt}^*(0,k)}{\overline{\gamma}}.
\]
The method of induction of the strategy $S$ is a straightforward generalization of the approach presented for reachability and B\"uchi automata.

\subsection{Parity conditions}\label{subparity}

The simple notions of rank and progress defined previously are insufficient to capture the complexity of parity acceptance conditions.
Instead we generalize progress measures for parity games \cite{klarlund90,Namjoshi01}. Note that, for clarity of exposition, all results in this section are presented for deterministic strategies only. The extension to non-deterministic strategies is straightforward.

Recall that Assumption \ref{connectparity} asks only that every state in $q$ is nominally coreachable with respect to some set of even parity $F_{2i}$, and if $q$ has odd parity, we assume that $2i$ is less than the parity of $q$. This is the least restrictive generalization of the coreachability assumptions made for simpler acceptance conditions. A consequence is that we extend the distance function $d$ to allow states of infinite distance from each other. Let $\overline{\mathbb R}_0^+ = \mathbb R_0^+ \cup \{\infty\}$, the extended positive reals. Then $d:Q \times Q \to \overline{\mathbb R}_0^+$ is an  \emph{extended distance function}.

Let $(A,\mathscr F) = ((Q, d), q_0, \Sigma, \delta,X, \gamma, \mathscr F)$ be a parity automaton with \\ $\mathscr F = \{F_0, F_1, \ldots, F_{2n+1} \}$. Denote by $d(q,\mathscr F)$
the vector valued distance
\[
d(q,\mathscr F) = ( d(q, F_0), d(q, F_2), d(q,F_4), \ldots, d(q, F_{2n})).
\]

Let $\succ$ denote the lexicographic ordering on $n+1$ tuples over the extended positive real numbers, and let $\succ^i$ denote the lexicographic ordering
restricted to the first $i$ components. We define $\succeq^i$ in the obvious way: $a \succeq^i b$ if $a$ is either greater than $b$ in the lexicographic ordering
or equal to $b$. For $a, b \in (\overline{\mathbb R}_0^+)^{n+1}$ define $a \rhd^q b$ if and only if there exists $i \in \{0,1, \ldots, n\}$ such that either
\begin{description}
\item[(i)] $q \in F_{2i +1}$ and $a \succ^i b$ or
\item[(ii)] $q \in F_{2i}$ and $a \succeq^i b$ or
\item[(iii)] $q \not\in \bigcup_{j \in \{0.\ldots, 2n+1\}} F_j$ and $a \succ b$.
\end{description}
We call $\rhd$ the \emph{parity progress measure}.

A (parity) rank function $R_{\mathscr F}:Q \to (\overline{\mathbb R}_0^+)^{n+1}$ is a function with $R_{\mathscr F}^i(q) = 0$ if and only if $q \in F_{2i}$ (where the notation
$R_{\mathscr F}^i(q)$ denotes the $i$th component of the image of $q$ under $R_{\mathscr F}$) and there exists a monotonically increasing function
$\alpha:\overline{\mathbb R}_0^+ \to \overline{\mathbb R}_0^+$ such that $\alpha(0) = 0$ and
\[ \alpha(d(q, F_i)) \leq R^i_{\mathscr F}(q) \]
for all $q \in Q$, $i \in \{0,\ldots, n\}$. Hence a parity rank function consists of $n+1$ reachability rank functions defined upon the extended positive real numbers.

Let $\overline{Q} \subseteq (F_0 \cup F_{2} \cup \ldots \cup F_{2n})$ denote the set of states of even parity from which a state of lower or equal even parity cannot be reached.
That is, $\overline{Q}$ contains all states $q \in F_{2i}$ for some $i \in \{0,\ldots,n\}$ such that there does not exist $k \leq i$ with some state $q' \in F_{2k}$ reachable from $q$.

A rank function $R_{\mathscr F}$ for a parity automaton $(A,\mathscr F)$ is a \emph{control Lyapunov function} if there exists some monotonically increasing function
$f: (\overline{\mathbb R}_0^+)^{n+1} \to (\overline{\mathbb R}_0^+)^{n+1}$ satisfying\footnote{where $0^{n+1}$ denotes the $n+1$ tuple consisting of zeroes.} $f(0^{n+1}) = 0^{n+1}$
such that for every $j \in \{1,\ldots, 2n+1\}$ and every $q \in F_{j} \setminus \overline{Q}$ there exists $a \in \Sigma$ with
\begin{equation} \label{parityineq}
R_{\mathscr F}(q^{a\epsilon}) - R_{\mathscr F}(q) \preceq^{i}  -f(d(q, \mathscr F))
\end{equation}
for some $2i \leq j$.

The next proposition demonstrates that the parity progress measure $\rhd$ is correct. Since the parity acceptance condition looks only at infinite behaviour on a trace, and we consider only automata with finite state sets, necessarily any infinite trace consists of a finite simple (loop-free) prefix followed by an infinite sequence of repeated loops. This observation is key to the proof.

\begin{proposition}\label{rhd}
Let $\tau = q_0 q_1 q_2 \ldots \in Q^{\omega}$ be an infinite trace of the parity automaton $(A,\mathscr F)$. Then if
\begin{equation}\label{rhdeq}
R_{\mathscr F}(q_0) \rhd^{q_0} R_{\mathscr F}(q_1) \rhd^{q_1} R_{\mathscr F}(q_2) \ldots,
\end{equation}
$\tau$ is winning with respect to $\mathscr F$.
Moreover, if the set of indices $I \subset \mathbb N$ such that
$R_{\mathscr F}(q_i) \rhd^{q_i} R_{\mathscr F}(q_{i+1})$ does not hold is finite, then $\tau$ will be winning with respect to $\mathscr F$.
\end{proposition}
\begin{proof}
Let $p_1, \ldots, p_m \in Q$ be such that
\begin{eqnarray}\label{rhdproof}
 R_{\mathscr F}(p_1) \rhd^{p_1} R_{\mathscr F}(p_2) \rhd^{p_2} \ldots \notag\\
 \ldots \rhd^{p_{m-2}} R_{\mathscr F}(p_{m-1}) \rhd^{p_{m-1}} R_{\mathscr F}(p_m) \rhd^{p_m} R_{\mathscr F}(p_1)
\end{eqnarray}
and let
\[ k = \min_{j \in \{1,\ldots,m\}} \{i \mid p_j \in F_{2i} \vee p_j \in F_{2i+1} \}.\]
By definition
\begin{equation}\label{rhdk}
 R_{\mathscr F}(p_1) \succeq^k R_{\mathscr F}(p_2) \succeq^k \ldots \succeq^k R_{\mathscr F}(p_m) \succeq^k R_{\mathscr F}(p_1).
\end{equation}
 Let $j \in \{1,\ldots, m\}$ be such that $p_j \in F_{2k+1}$ or $p_j \in F_{2k}$. If $p_j \in F_{2k+1}$ then one of the inequalities
in (\ref{rhdk}) must be strict and hence $R_{\mathscr F}(p_1) \succ^k R_{\mathscr F}(p_1)$, a contradiction. Therefore $p_j \in F_{2k}$ and the least parity appearing in the loop
$p_1 \ldots p_m p_1$ must be even.
This is sufficient to prove that any infinite trace $\tau \in Q^{\omega}$ satisfying (\ref{rhdeq}) also satisfies the parity
condition $\mathscr F$.

Now assume that the set $I$ is non-empty and has finite cardinality. Since $I$ is finite there exists some $N \in \mathbb N$ such that for all $k \geq N$, $R_{\mathscr F}(q_k) \rhd^{q_k} R_{\mathscr F}(q_{k+1})$ holds. Let $\tau_N$ denote
the suffix of $\tau$ whose first state is $q_N$. Then by Proposition \ref{rhd} the lowest parity in the set $\zeta(\tau_N)$ is even, and since $\zeta(\tau) = \zeta(\tau_N)$ the result follows.
\end{proof}

As we observed before the proposition, a nominally winning infinite trace of a finite state parity automaton is necessarily comprised of a finite simple prefix followed by an infinite series of repeated loops. It is then straightforward to argue that the least parity appearing on any such loop must be even. Continuing on this line of thinking one observes that any such repeated loop comprising part of an infinite trace satisfying (\ref{rhdeq}) must consist entirely of even states. Hence a trace of this form will feature odd states only finitely often.

\begin{proposition}\label{CLFtorhd}
 Let $\tau = q_0 q_1 q_2 \ldots \in Q^{\omega}$ be an infinite trace of the parity automaton $(A,\mathscr F)$. If
\begin{equation} \label{parityineq2}
R_{\mathscr F}(q_{k+1}) - R_{\mathscr F}(q_k) \preceq^{i}  -f(d(q_k, \mathscr F))
\end{equation}
for all $ q_k \in Q \setminus \bar{Q}$ appearing on $\tau$ and $\bar{Q}$ is finite then $\tau$ satisfies $\mathscr F$.
\end{proposition}
\begin{proof}
Let $q_k \in Q \setminus \bar{Q}$. If $q_k \in F_{2i}$ for some $i$ then  (\ref{parityineq2}) implies that
$R_{\mathscr F}(q_k) \succeq^i R_{\mathscr F}(q_{k+1})$ and $q_{k} \rhd^{q_k} q_{k+1}$ as required.

Instead assume that $q_k \in F_{j}$ for some $j$ odd. The function $f(d(q_k,\mathscr F))$ restricted to any \mbox{$i \in \{0,\ldots, n\}$} is non-zero, and so
$R_{\mathscr F}(q_{k+1}) \prec^i R_{\mathscr F}(q_k)$ for some $2i<j$ and hence for all $l$ satisfying $i \leq l < j$. Therefore $R_{\mathscr F}(q_{k}) \rhd^{q_k} R_{\mathscr F}(q_{k+1})$.

Finally let $q_k \in \bar{Q}$ and $q_k \in F_{2i}$. Then $q_k$ and $q_{k+1}$ need not satisfy (\ref{parityineq2})
and so may not satisfy the parity measure $\rhd$.
Since $q \in \bar{Q}$ there exists no $l > k$ such that $q_l = q_k$. Indeed, if this were the case, it would contradict
our assumption that a finite trace connecting $q_k$ to a state of lower or equal parity does not exist.
Since the cardinality of $\bar{Q}$ is finite there exist only a finite number
of indices $l \in \mathbb N$ such that $R_{\mathscr F}(q_l) \rhd^{q_l} R_{\mathscr F}(q_{l+1})$ does not hold
and Proposition \ref{rhd} yields the result.
\end{proof}


Given a control Lyapunov function $R_{\mathscr F}$ for a parity automaton $(A,\mathscr F)$ a deterministic memoryless strategy $S:Q \to \Sigma$ induced from $R_{\mathscr F}$
may be defined as follows. Let $q \in Q$.
\begin{description}
\item[(i)] If $q \in Q \setminus \bar{Q}$ choose $S(q) = a$ such that $R_{\mathscr F}(q^{a\epsilon})$ satisfies (\ref{parityineq2}) and is minimal with respect
 to the lexicographic ordering.
\item[(ii)] If $q \in \bar{Q}$ set $S(q) = a$ for any $a \in \Sigma$.
\end{description}

\begin{theorem}\label{parity}
Let $(A,\mathscr F)$ be a finite parity automaton satisfying Assumption \ref{connectparity} and let $S:Q \to \Sigma$ be a deterministic memoryless strategy. Then $S$ is nominally winning if and only if there exists a Lipschitz continuous control Lyapunov function $R_{\mathscr F}$ such that $S$ may be induced from $R_{\mathscr F}$.
\end{theorem}
\begin{proof}
That a strategy induced from a control Lyapunov function is nominally winning follows immediately from Proposition \ref{CLFtorhd}.
So let $S:Q \to \Sigma$ be a deterministic memoryless nominally winning strategy for $(A,\mathscr F)$.
In order to synthesize a control Lyapunov function from which $S$ may be induced, the state set $Q$ is partitioned into $n+1$ pieces,
\[ Q = \overline{F_0} \cup \overline{F_2} \cup \ldots \cup \overline{F_{2n}} \]
where the sets $\overline{F_{2i}}$ for $i = 0,\ldots, n$ are defined as follows.
For $q \in Q$, let $i \in \{0,\ldots,n\}$ be the least such that there exists a trace resulting from applying $S$ in $A$ connecting $q$ to a state in the set $F_{2i}$. Then the state $q$ is contained in $\overline{F_{2i}}$. Since we assume that a state of even parity may be reached from all states in $Q$, the resulting sets form a partition.

We construct from $(A,\mathscr F)$ a weighted digraph $(Q,E)$. An edge $(q,x_q,q')$ is contained in the edge set $E$ if and only if
$q \in \overline{F_{2i}} \setminus F_{2i}$ for some $i \in \{0,\ldots,n\}$ and $\delta(q,S(q),\epsilon)=q'$. Let \mbox{$\eta:\mathbb R_0^+ \to \mathbb R_0^+$} be a monotonically increasing function with $\eta(0) = 0$. The value \mbox{$x_q=(x_0,x_1,\ldots,x_n) \in (\mathbb R_0^+)^{n+1}$} is defined as follows:
\begin{itemize}
\item for all $j \in \{0,\ldots,n\}$ with $j \geq i$, $x_j =  \eta(d(q,F_{2i}))$ ;
\item for all $j \in \{0,\ldots,n\}$ with $j < i$, $x_j = \infty$.
\end{itemize}

Define  $R:Q \to (\overline{\mathbb R}_0^+)^{n+1}$ where for $q \in \overline{F_{2i}}\setminus F_{2i}$, $R(q) = \sum_{q' \in \tau} x_{q'}$ where $\tau \in Q^*$ is the unique trace connecting $q$ to some state in $F_{2i}$ resulting from applying the strategy $S$ in $A$.

Let $F_d = \bigcup_{i=0}^n (F_{2i} \cap \overline{F_{2i}})$, the set of states for which the function $R$ has not been defined. Notice that it is not necessarily the case that $F_d = F$. These states are precisely those states of even parity from which a state of lower or equal even parity cannot be reached - that is, $F_d$ coincides precisely with the set $\overline{Q}$. For $q \in \overline{Q}$ set $R^i(q) = 0$ where $q \in F_{2i}$, $R^j(q) = d(q,F_{2j})$ for $j > i$ and $R^j(q) = \infty$ for $j < i$ where $j \in \{0,\ldots,n\}$.

We once again observe that for all $q \in Q \setminus \overline{Q}$
\[ R(\delta(q,S(q),\epsilon)) - R(q) \prec^i -\eta(d(q,\mathscr F)) \]
for some $i$ depending on $q$. Hence $R$ is a control Lyapunov function.

Since the choice of input for $q \in \overline{Q}$ may be arbitrary for a strategy induced from $R$, the result follows.
\end{proof}

 The following result takes advantage of the extra flexibility resulting from a partial colouring of the state set. If each set $F_{2i}$ for $i = 0, \ldots, n$
has only non-parity states in its immediate neighbourhood, the sets may be inflated without overlap to ensure that a strategy induced from a control Lyapunov function
is winning for an inflated acceptance condition $\mathscr F'$ as defined below.

\begin{theorem}\label{paritysigma}
Let $(A,\mathscr F)$ be a finite parity automaton satisfying Assumption \ref{connectparity} with constant disturbance bound $\overline{\gamma}$ and let $S$ be a deterministic memoryless strategy induced from a control Lyapunov function $R_{\mathscr F}$.
Further, let $F = \bigcup_{j=0}^{2n+1} F_j$ be such that $F \subsetneq Q$ and for all
$i \in \{0, \ldots , n\}$ if $q \not\in F_{2i}$ and $f(d(q, F_{2i})) \leq K\overline{\gamma}$ then $q \not\in F$.
Then $S$ is a $\sigma$-robust winning strategy for $\sigma = f^{-1}(K\overline{\gamma})/\overline{\gamma}$.

\end{theorem}
\begin{proof}
Assume first that $R_{\mathscr F}$ is a Lipschitz continuous control Lyapunov function for $(A,\mathscr F)$ and let $S$ be a deterministic memoryless strategy induced from $R_{\mathscr F}$. Let $T$ be a disturbance strategy and let $\tau \in Q^{\omega}$ be the unique nominal outcome resulting from $S$ and $T$.
An argument similar to the one used in Theorem \ref{safety} implies that for each $q \in \tau$
\begin{equation}\label{parineq}
R_{\mathscr F}(q^{ax}) - R_{\mathscr F}(q) \preceq^i (K\gamma(q^{ax}),\hdots,K\gamma(q^{ax})) -f(d(q,\mathscr F)).
\end{equation}
Let $2i$ be the least colour appearing infinitely often on $\tau$ and define $F_{2i}' = \{ q \in Q \mid d(q,F_{2i}) \leq \sigma \overline{\gamma}\}$ for $i=0,\ldots,n$.
Inequality (\ref{parineq}) implies that $\tau$ will visit infinitely often states in the set $F'_{2i}$ in $A$.
Since, by assumption, states in $F_{2i}' \setminus F_{2i}$ are not contained in $F$, the inflation from $F_{2i}$ to $F'_{2i}$ will not cause any state to have more than one parity,
and we conclude that the strategy $S$ is $\sigma$-robust.
\end{proof}

For parity automata with state dependent disturbance bounds we again use the operators $\overline{g}$ and $\underline{g}$, but this time with some modifications to take advantage of the progress measure $\rhd$. As for the case of reachability automata, $opt^0$ is defined to be a vector of size $m$ over the positive reals, however this time we let $opt^0(j) = ( d(q_j,F_0), d(q_j,F_2), \ldots, d(q_j,F_{2n}))$ where $\mathscr F = \{ F_0,F_1,\ldots,F_{2n+1}\}$. The operators $\overline{g}$ and $\underline{g}$ are defined in the same way as for reachability automata but the underlying ordering used for the minimum operation is the lexicographic ordering on the $n+1$-tuples instead of numerical ordering as in previous cases. This alteration will not affect the complexity of the algorithm.

 For a nominally winning strategy $S$ for a finite parity automaton $(A,\mathscr F)$, $\sigma$ may be recovered by calculating $\underline{opt}^*$ for $A|_{S}$. We abuse previous notation and let $opt(j,k)$ denote the $k$-th index in the $n+1$ tuple appearing on the $j$th line of the vector $opt$.  Then
\[ \sigma= \frac{\max_{k = 1, \ldots, n+1}\overline{opt}^*(0,k)}{\overline{\gamma}}.\]

Then if $F = \bigcup_{k=0}^{2n+1} F_k$ is such that $F \subsetneq Q$ and for all
$i \in \{0, \ldots , n\}$ if $q \not\in F_{2i}$ and \mbox{$d(q, F_{2i}) \leq \sigma\overline{\gamma}$} then $q \not\in F$, the strategy $S$ is $\sigma$-robust.

For optimal strategy synthesis we first restrict the automaton $A$ with respect to the progress measure $\rhd$. For all $q \in Q$, $a \in \Lambda(q)$ if and only if $opt_0(j) \rhd opt_0(j')$ where $\delta(q,a,\epsilon)=q_{j'}$. We denote the automaton restricted in this way by $A|_{\rhd}$. Calculating $\underline{opt}^*$ for $A|_{\rhd}$, the optimum achievable value of $\sigma$ is recovered as
\[ \sigma_{min} = \frac{\min_{k=1,\ldots,n+1} \underline{opt}^*(0,k)}{\overline{\gamma}}.\]
The method of induction of the strategy $S$ is a straightforward generalization of the approach presented for previous acceptance conditions. Again one must check the separation of the even parity sets with respect to the distance function $d$ to ensure that the resulting strategy will be robust.

\section{Application: Transient faults}

Transient faults, such as single-event upsets, are unpredictable disturbances in electronic systems
that can cause bits in an electronic circuit to flip.
They are becoming more relevant in electronic systems design due to reductions in feature sizes \cite{Borkar:DAC06,rosetta:xilinx,normand:ieee}.
We show that strategies synthesized using control Lyapunov functions are robust to infinitely occurring transient
faults provided they occur infrequently enough.

Let $N\in\N$.
A disturbance strategy $T: Q^+\times \Sigma \to X$
is {\em $N$-bounded} if, whenever $T(\tau,a) \not = \epsilon$
and $T(\tau',b) \not = \epsilon$ for traces $\tau,\tau'\in Q^*$ with $\tau$ a proper prefix
of $\tau'$ and $a,b\in\Sigma$, we have \mbox{$|\tau'| - |\tau| \geq N$.}
Intuitively, disturbance strategies are $N$-bounded if any two occurrences of (non-trivial)
disturbances are separated by at least $N$ steps.

Let $A$ be an automaton and $F$ a (B\"uchi or parity)
acceptance condition.
Our main result is that for sufficiently large (but finite) $N$,
a nominally winning strategy induced from a control Lyapunov function
is winning against $N$-bounded disturbance strategies.


\begin{proposition}\label{buchiinffault}
Let $A$ be an automaton, $F$ a B\"uchi acceptance condition and $\mathscr F$ a
parity acceptance condition.
\begin{description}
\item[(i)]
Let $R_{F}$ be a control Lyapunov function for the B\"uchi automaton $(A,F)$
and let $S$ be a $\sigma$-robust deterministic strategy induced from $R_{F}$.
Then $S$ is winning against every $N$-bounded disturbance strategy
with $N\geq \max_{q\in F'} |\tau^S(q,F)|$
where
$F' = \{ q \in Q \mid d(q,F) \leq \sigma \overline{\gamma} \}$.

\item[(ii)]
Let $R_{\mathscr F}$ be a control Lyapunov function for the parity automaton $(A,\mathscr F)$
and let $S$ be a $\sigma$-robust strategy induced from $R_{\mathscr F}$.
Then $S$ is winning against every $N$-bounded disturbance strategy,
where
\[
\infty > N \geq \max_{i = 0, \ldots, n}\left(\max \left( \{|\tau^S(q,F_{2i})|: q \in F'_{2i}\} \cap \mathbb R \right) \right)
\]
for $F'_{2i} = \{ q \in Q \mid d(q,F_{2i}) \leq \sigma \overline{\gamma} \}$.

\end{description}
%
\end{proposition}

In (ii) it is important that the value of $N$ is finite. Indeed, that $N$ is not finite is a possibility since there may exist sets of even parity which are not reachable from a given state $q \in Q$.

\begin{proof}
For (i), we show that for any $q \in Q$ there exists a finite trace in $A$ connecting $q$ to $F$ resulting from applying $S$.

First let $q \in Q$ be such that $d(q,F) > \sigma \overline{\gamma}$. Then since $S$ is $\sigma$-robust there exists a unique finite trace $\tau^S(q,F)$  ending at some state $q' \in Q$ such that $d(q',F)) \leq \sigma \overline{\gamma}$, regardless of how frequently the fault occurs.

Now assume $q \in Q$ is such that $d(q,F) \leq \sigma \overline{\gamma}$. By assumption if the unique trace $\tau^S(q_0,\{q\})$ is such that $q = p^{ax}$ for some $x \neq \epsilon$, that is,
the state $q$ was reached due to the effects of a fault, the next $N$ transitions on the trace will be nominal, that is, $x = \epsilon$. By definition of $N$ and $S$
the resulting subtrace of length $N$ will visit a state in the set $F$.

Assume instead that $q = p^{a\epsilon} \not\in F$. If the next state $q^{a'x}$ on the trace resulting from $S$ is such that $x = \epsilon$ then $q^{a'x}$ will satisfy $d(q^{a'x})\leq \sigma \overline{\gamma}$
and the same argument may be applied. Therefore if no fault occurs for the next $N$ transitions some state in the set $F$ will be reached. If a fault occurs, a state in the set $F$ will be reached in the $N$ transitions
following the fault.

If $x \neq \epsilon$
then either $d(q^{ax},F)) \leq \sigma \overline{\gamma}$ in which case the above argument applies, or \mbox{$d(q^{ax},F) > \sigma \overline{\gamma}$}
and the first argument applies. So we conclude that the strategy $S$ is winning in the automaton $A$ against an $N$-bounded disturbance strategy.


For (ii) the argument is similar. If $q$ has even parity then the result follows.
Assume instead that $q \in F_{j}$ with $j$ odd.
If for $i \in \{0,\ldots,n\}$, $\infty > d(q,F_{2i}) > \sigma \overline{\gamma}$ with $2i<j$ then since $S$ is $\sigma$-robust
 there exists a finite trace resulting from $S$ connecting $q$ to some state $q'$ satisfying $d(q',F_{2k}) < \sigma \overline{\gamma}$ for some $k \in \{0,\ldots,n\}$ regardless of how frequently the fault occurs.

Now assume $q$ and $i \in \{0,\ldots, n\}$ are such that $2i<j$ and $d(q,F_{2i}) \leq \sigma \overline{\gamma}$.
By assumption if the trace resulting from $S$ connecting $q_0$ to $q$ is such that $q = p^{ax}$ for some
$p \in Q$, $a\in \Sigma$ and $ x\in X \setminus \{\epsilon\}$ the next $N$ transitions will be nominal and the resulting subtrace will feature a state in the set $F_{2i}$.
If instead $q = p^{a\epsilon}$ then either
\begin{description}
\item[(i)] the next state on the trace is contained in a set $F_{2i}$ for some $2i<j$ and we are done;
\item[(ii)] $q^{ax} \not\in F_{2i}$ for some $i$ and $x \neq \epsilon$ and a state in a set of lower even parity will be reached in the next $N$ steps or
\item[(iii)] the next state $q^{ax}$ is such that $x=\epsilon$. Then the argument is repeated: if a fault does not occur for the next $N$ transitions then a state in a set
of lower even parity will be visited. If a fault occurs, a state of even parity will be visited in the next $N$ transitions following the fault.
\end{description}

Therefore a strategy $S$
induced from a Lipschitz continuous control Lyapunov function is winning for the parity automaton $(A,\mathscr F)$ against an $N$-bounded disturbance strategy.
\end{proof}

Compare the above result to the equivalent bound one might establish for a strategy induced from a classical shortest path rank function in a B\"uchi automaton:
in this case the value of $N$ must be greater than the length of the longest simple path connecting a state in $Q$ to a state in $F$.
In our result $N$ is defined with respect
to a potentially much smaller subset of $Q$.
Since the bound $N$ is a monotonically increasing function
of the environmental error $\overline{\gamma}$ this result provides a bridge between the state based view of faults and the running time of the system:
a less powerful fault may occur more frequently than a more powerful one without disrupting a well designed strategy.

\section{Discussion}
\label{Sec:Discussions}

We have presented a theory of robustness for $\omega$-regular properties
of automata. We have considered both deterministic and non-deterministic memoryless strategies, and disturbances whose power is bounded universally across the whole system, or bounded dependent upon the current state. In every case we provide methods to explicitly calculate and guarantee robustness of given strategies, as well as polynomial time algorithms to synthesize optimally robust strategies for a given system.
%
There are two natural extensions to our work.
First, in our model, bounded disturbances are the only source of adversarial interaction.
The presence of additional adversaries leads to (more complex) algorithms for solving two-player
games \cite{Bloem10,Zielonka98}.
We believe our simpler model is already applicable in many settings ---we are inspired by similar
models in continuous control--- and our polynomial-time algorithms render our results
applicable in practice. It would therefore be of interest to see how our results extend to a setting in which additional adversarial influences exist. 

Second, how can we combine our results on automata with the existing theory of robust control
for continuous systems? We believe that by consolidating some of the recently reported results~\cite{PGT08,ZPT10} on the existence of automata based abstractions of continuous control systems with the methods presented here we can expect to obtain a comprehensive robustness theory for cyber physical systems.

\bibliographystyle{plain}
\bibliography{ref}

\begin{thebibliography}{10}

\bibitem{Arora93}
A.~Arora and M.~G. Gouda.
\newblock Closure and convergence: a foundation of fault tolerant computing.
\newblock {\em IEEE Transactions on Software Engineering}, 19(11):1015--1027,
  1993.

\bibitem{Bellman54}
R.E. Bellman.
\newblock The theory of dynamic programming.
\newblock {\em Bull.\ Amer.\ Math.\ Soc.}, 60:503--516, 1954.

\bibitem{Bloem10}
R.~Bloem, K.~Chatterjee, K.~Greimel, T.A. Henzinger, and B.~Jobstmann.
\newblock Robustness in the presence of liveness.
\newblock In {\em CAV 2010}, volume 6174 of {\em Lecture Notes in Computer
  Science}, pages 410--424. Springer, 2010.

\bibitem{Bloem091}
R.~Bloem, K.~Chatterjee, T.~A. Henzinger, and B.~Jobstmann.
\newblock Better quality in synthesis through quantitative objectives.
\newblock In {\em CAV 2009: Computer-Aided Verification}, Lecture Notes in
  Computer Science 5643, pages 140--156. Springer-Verlag, 2009.

\bibitem{Bloem09}
R.~Bloem, K.~Greimel, T.A. Henzinger, and B.~Jobstmann.
\newblock Synthesizing robust systems.
\newblock In {\em FMCAD 09: Formal Methods in Computer-Aided Design}, pages
  85--92. IEEE, 2009.

\bibitem{Borkar:DAC06}
S.~Borkar.
\newblock {Electronics Beyond Nano-scale CMOS}.
\newblock In {\em DAC 06}. ACM, 2006.

\bibitem{HSTopology}
M.S. Branicky.
\newblock Topology of hybrid systems.
\newblock In {\em Proceedings of the 32nd IEEE Conference on Decision and
  Control}, pages 2309--2314, 1993.

\bibitem{Cerny10}
P.~Cern\'y, T.~A. Henzinger, and A.~Radhakrishna.
\newblock Simulation distances.
\newblock In {\em CONCUR 2010 - Concurrency Theory}, volume 6269 of {\em
  Lecture Notes in Computer Science 6269}, pages 253--268. Springer-Verlag,
  2010.

\bibitem{Dijkstra74}
E.~W. Dijkstra.
\newblock Self-stabilizing systems in spite of distributed control.
\newblock {\em Communications of the ACM}, 17(11):643--644, 1974.

\bibitem{EmersonJutla91}
E.A. Emerson and C.~Jutla.
\newblock Tree automata, mu-calculus and determinacy.
\newblock In {\em Proceedings of the 32th Annual Symposium on Foundations of
  Computer Science}, pages 368--377. IEEE Computer Society Press, 1991.

\bibitem{Girault09}
A.~Girault and E.~Rutten.
\newblock Automating the addition of fault tolerance with discrete controller
  synthesis.
\newblock {\em Formal Methods in System Design}, 35(2):190--225, 2009.

\bibitem{golshan07:dac}
S.~Golshan and E.~Bozorgzadeh.
\newblock {Single-Event-Upset ({SEU}) Awareness in {FPGA} Routing}.
\newblock In {\em DAC 07}. ACM, 2007.

\bibitem{Hu-iccad08}
Y.~Hu, Z~.Feng, L.~He, and R.~Majumdar.
\newblock Robust {FPGA} resynthesis based on fault-tolerant boolean matching.
\newblock In {\em ICCAD}, Nov 2008.

\bibitem{klarlund90}
N.~Klarlund.
\newblock {\em Progress measures and finite arguments for infinite
  computations}.
\newblock PhD thesis, Cornell University, 1990.

\bibitem{Krishnaswamy-tcad09}
S.~Krishnaswamy, S.~Plaza, I.~Markov, and J.~Hayes.
\newblock Signature-based {SER} analysis and design of logic circuits.
\newblock {\em TCAD}, 2009.

\bibitem{rosetta:xilinx}
A.~Lesea, S.~Drimer, J.J. Fabula, C.~Carmichael, and P.~Alfke.
\newblock {The {R}osetta experiment: atmospheric soft error rate testing in
  differing technology {FPGAs}}.
\newblock {\em IEEE Transactions on Device and Materials Reliability},
  5(3):317--328, 2005.

\bibitem{Lynch}
N.A. Lynch.
\newblock {\em Distributed algorithms}.
\newblock Morgan Kaufmann, 1996.

\bibitem{McNaughton}
R.~McNaughton.
\newblock Infinite gam,es played on finite graphs.
\newblock {\em Annals of Pure and Applied Logic}, 65(2):149--184, 1993.

\bibitem{Miskov-ZivanovM10}
N.~Miskov-Zivanov and D.~Marculescu.
\newblock Formal modeling and reasoning for reliability analysis.
\newblock In {\em DAC}, pages 531--536. ACM, 2010.

\bibitem{Namjoshi01}
K.~Namjoshi.
\newblock Certifying model checkers.
\newblock In {\em CAV 01: Computer Aided Verification}, volume 2102 of {\em
  Lecture Notes in Computer Science 2102}, pages 2--13. Springer-Verlag, 2001.

\bibitem{HSModels}
A.~Nerode and W.~Kohn.
\newblock Models for hybrid systems: Automata, topologies, controllability,
  observability.
\newblock In {\em Hybrid Systems}, Lecture Notes in Computer Science 736, pages
  297--316. Springer-Verlag, 1993.

\bibitem{normand:ieee}
E.~Normand.
\newblock Single event upset at ground level.
\newblock {\em IEEE Transactions on Nuclear Science}, 43(6):2742--2750, 1996.

\bibitem{PGT08}
G.~Pola, A.~Girard, and P.~Tabuada.
\newblock Approximately bisimilar symbolic models for nonlinear control
  systems.
\newblock {\em Automatica}, 44(10):2508--2516, 2008.

\bibitem{Tarraf08}
D.C. Tarraf, A.~Megretski, and M.A. Dahleh.
\newblock A framework for robust stability of systems over finite alphabets.
\newblock {\em IEEE Transactions on Automatic Control}, 53(5):1133--1146, 2008.

\bibitem{Thomas95}
W.~Thomas.
\newblock On the synthesis of strategies in infinite games.
\newblock In {\em STACS 95: Theoretical Aspects of Computer Science}, volume
  900 of {\em Lecture Notes in Computer Science}, pages 1--13. Springer-Verlag,
  1995.

\bibitem{L2Gain}
A.J. van~der Schaft.
\newblock {\em L2-Gain and Passivity Techniques in Nonlinear Control}, volume
  218 of {\em Lecture Notes in Control and Information Sciences}.
\newblock Springer-Verlag, 2000.

\bibitem{DigitalDesignTextbook}
J.F. Wakerly.
\newblock {\em Digital Design Principles and Practices}.
\newblock Prentice Hall, 1994.

\bibitem{ZPT10}
M.~Zamani, G.~Pola, and Paulo Tabuada.
\newblock Symbolic models for unstable nonlinear control systems.
\newblock In {\em Proceedings of the 2010 American Control Conference}, 2010.

\bibitem{ROBUSTCONTROL_TEXTBOOK}
K.~Zhou, J.~Doyle, and K.~Glover.
\newblock {\em Robust and Optimal Control}.
\newblock Prentice Hall, 1996.

\bibitem{Zielonka98}
W.~Zielonka.
\newblock Infinite games on finitely coloured graphs with applications to
  automata on infinite trees.
\newblock {\em Theor. Comput. Sci.}, 200(1-2):135--183, 1998.

\end{thebibliography}


\end{document}